\documentclass{article}
\usepackage{enumitem}
\setlist[itemize]{leftmargin=*}
\setlist[enumerate]{leftmargin=*}
    \PassOptionsToPackage{numbers, compress}{natbib}

\usepackage{subcaption}
\newcommand{\Prob}[1]{Pr[#1]}
\renewcommand{\epsilon}{\varepsilon}

\usepackage[preprint]{neurips_2022}
\usepackage{graphicx}
\usepackage{wrapfig}
\usepackage{amsthm}
\usepackage{amsmath}
\usepackage{amssymb}

\theoremstyle{definition}
\newtheorem{dfn}{Definition}
\newtheorem{thm}{Theorem}
\newtheorem{proposition}{Proposition}
\usepackage{algorithmic}
\usepackage{algorithm}
\usepackage{booktabs}
\usepackage{multirow}

\usepackage{tikz}
\usetikzlibrary{calc}

\usepackage[utf8]{inputenc} 
\usepackage[T1]{fontenc}    
\usepackage{hyperref}       
\usepackage{url}            
\usepackage{booktabs}       
\usepackage{amsfonts}       
\usepackage{nicefrac}       
\usepackage{microtype}      
\usepackage{xcolor}         

\title{\fontsize{17pt}{0cm}\selectfont{Measuring Lower Bounds of Local Differential Privacy via Adversary Instantiations in Federated Learning}}

\author{
  Marin Matsumoto\\
  Ochanomizu University\\
  \texttt{marin@ogl.is.ocha.ac.jp} \\
  \And
  Tsubasa Takahashi \\
  LINE Corporation \\
  \texttt{tsubasa.takahashi@linecorp.com} \\
  \And
  Seng Pei Liew \\
  LINE Corporation \\
  \texttt{sengpei.liew@linecorp.com} \\
  \And
  Masato Oguchi \\
  Ochanomizu University \\
  \texttt{oguchi@is.ocha.ac.jp} \\
}

\begin{document}

\maketitle

\begin{abstract}
Local differential privacy (LDP) gives a strong privacy guarantee to be used in a distributed setting like federated learning (FL).
LDP mechanisms in FL protect a client's gradient by randomizing it on the client; however, how can we interpret the privacy level given by the randomization?
Moreover, what types of attacks can we mitigate in practice?
To answer these questions, we introduce an empirical privacy test by measuring the lower bounds of LDP.
The privacy test estimates how an adversary predicts if a reported randomized gradient was crafted from a raw gradient $g_1$ or $g_2$.
We then instantiate six adversaries in FL under LDP to measure empirical LDP at various attack surfaces, including a worst-case attack that reaches the theoretical upper bound of LDP.
The empirical privacy test with the adversary instantiations enables us to interpret LDP more intuitively and discuss relaxation of the privacy parameter until a particular instantiated attack surfaces.
We also demonstrate numerical observations of the measured privacy in these adversarial settings, and the worst-case attack is not realistic in FL.
In the end, we also discuss the possible relaxation of privacy levels in FL under LDP.
\end{abstract}

\section{Introduction}
Federated learning (FL) is a machine learning technique in which clients never share the raw data but the gradient with the server.
Exposing only the gradient would seem to protect the client's privacy; however, several studies\citep{paper-inverting-grad, paper-cvpr-gradinversion} have shown that the gradient can recover the original image.
One way of privacy protection is applying local differential privacy (LDP) \citep{evfimievski2003limiting, kasiviswanathan2011can}.
LDP has been a widely accepted privacy standard that makes it hard to distinguish randomized responses crafted from any inputs to the extent that the privacy parameter $\epsilon$ quantifies it.
Recent studies have been seeking a practical LDP mechanisms in FL \citep{bhowmick2018protection,liu2020fedsel,girgis2021shuffled}.

To seek practice of LDP in FL, we need to understand the criteria for setting the privacy parameter $\epsilon$.
We have a study about statistical interpretation \citep{hoshino2020firmfoundation,lee2011much}, but the interpretations are still complicated and not enough for practitioners yet.
To give more intuitive interpretations that enable us to assess guaranteed privacy, we have the following two issues;
1) how can we interpret the privacy level given by LDP mechanisms?
2) what types of attacks can we mitigate in practice?
To answer these questions, we first introduce a privacy measurement test based on statistical discrimination between randomized responses under various adversaries.
Based on the privacy measurement test, we tackle instantiations of adversaries in FL using locally differentially private stochastic gradient descent (LDP-SGD) \cite{duchi2018,ldp-sgd-google}.
As for the adversary instantiations, we consider various capabilities for crafting raw gradients (i.e., what types of information the adversary is allowed to access).

Figure \ref{fig:overview} shows our empirical privacy measurement that is inspired by the hypothesis tests for measuring the lower bound of central DP (CDP) \cite{adversary-instantiation, auditing-DP}.
The privacy test estimates how an adversary predicts if a reported randomized gradient was crafted from a raw gradient $g_1$ or $g_2$.
To realize such distinguishing game, we introduce \textbf{crafter} and \textbf{distinguisher}.
Crafter (maliciously) generates $g_1$ and $g_2$, then submits a randomized gradient of either one.
Distinguisher attempts to discriminate the source of the randomized gradient.
The empirical LDP is estimated over sufficient number of games.

As for the adversary instantiations, the main challenge is to actualize robust adversary attacks against the noise to achieve $\epsilon$-LDP since the LDP mechanisms tend to inject exetemely larger noise than those of CDP.
In addition, we also need to consider collusions between the server and the client.
This paper tackles these challenges to realize the way to measure the lower bounds of LDP in FL, and demonstrates the measured $\epsilon$ ranging from the benign to the worst-case in Figure \ref{fig:summary}.

As a result of our solutions, this paper shows the following four facts:
(a) adversaries who directly manipulate raw gradients (i.e., \textit{white-box} distinguisher) can reach the theoretical upper bounds given by the privacy parameter $\epsilon$,
(b) even in white-box scenarios, the crafters with accessing only the input data and parameters but not the raw gradient have limited capability,
(c) the distinguishers with referring only to the updated model but not any randomized gradient (i.e., \textit{black-box} distinguisher) have limited capability,
(d) collusion with the server, like malicious pre-training, increases the adversary's capability.
See the details in section \ref{sec-exp}.

\textbf{Contributions.}
We here summarize our contributions as following three claims:
\begin{enumerate}[nosep]
    \item We establish the empirical privacy measurement test of LDP in distributed learning environments.
    \item We introduce six types of adversary instantiations that helps practitioners interpret the privacy parameter $\epsilon$. We also seek the worst-case attack in FL using LDP-SGD.
    \item We demonstrate the empirical privacy levels (i.e., lower bounds of LDP) at various attack surfaces defined by the instantiated adversaries.
\end{enumerate}
This study enables us to discuss appropriate privacy levels and their relaxations between the (unrealistic) worst-case scenario and the other weak adversaries.
At the end, we also discuss the feasibility of the relaxation with introducing a trusted entity like secure aggregator.

\begin{figure}[t]
     \centering
     \begin{subfigure}[t]{.6\linewidth}
     \includegraphics[width=\textwidth]{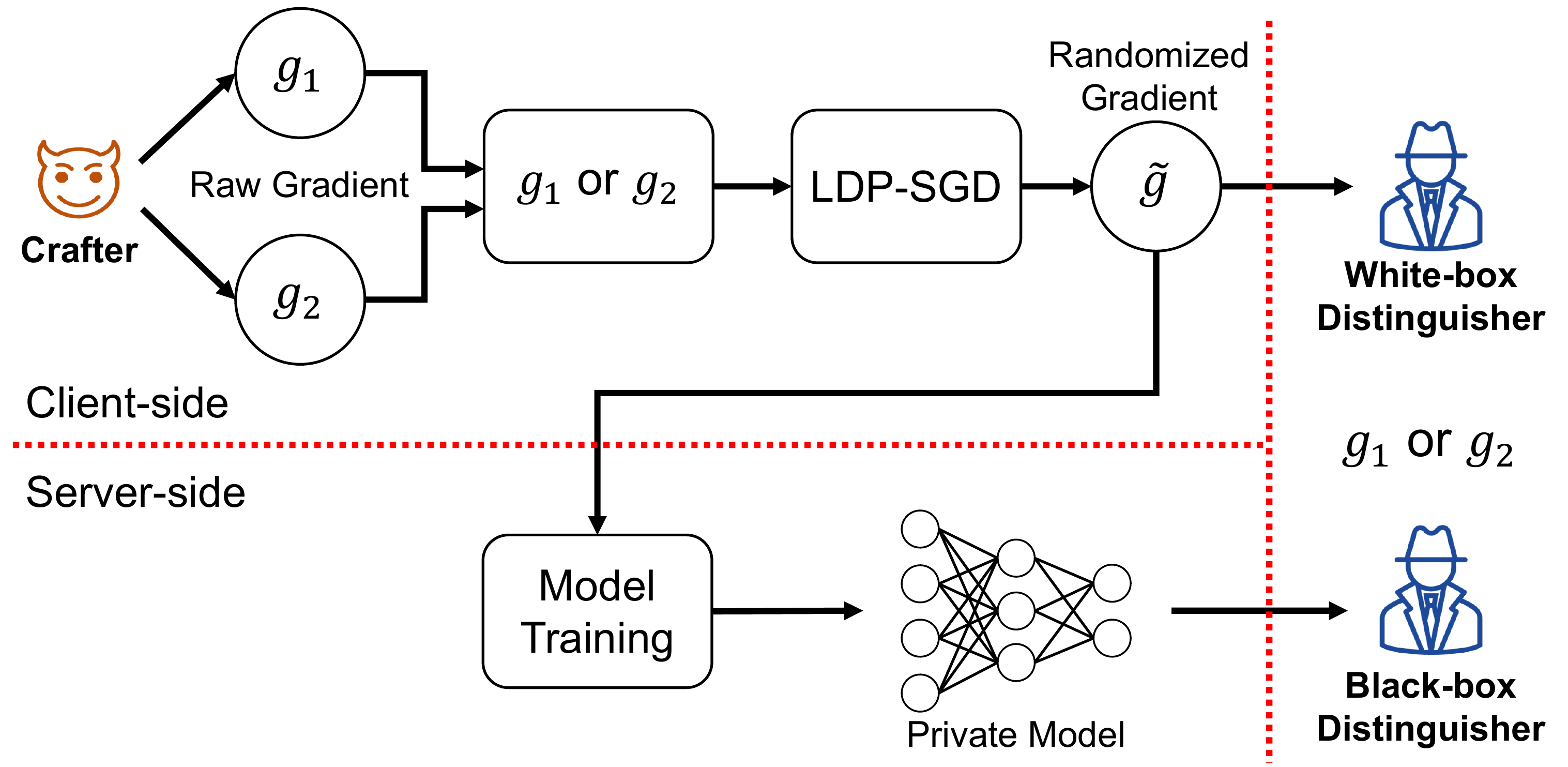}
     \caption{Our privacy measuring test in federated learning.}
     \label{fig:overview}
     \end{subfigure}
     \hfill
     \begin{subfigure}[t]{.35\linewidth}
         \centering
     \includegraphics[width=\textwidth]{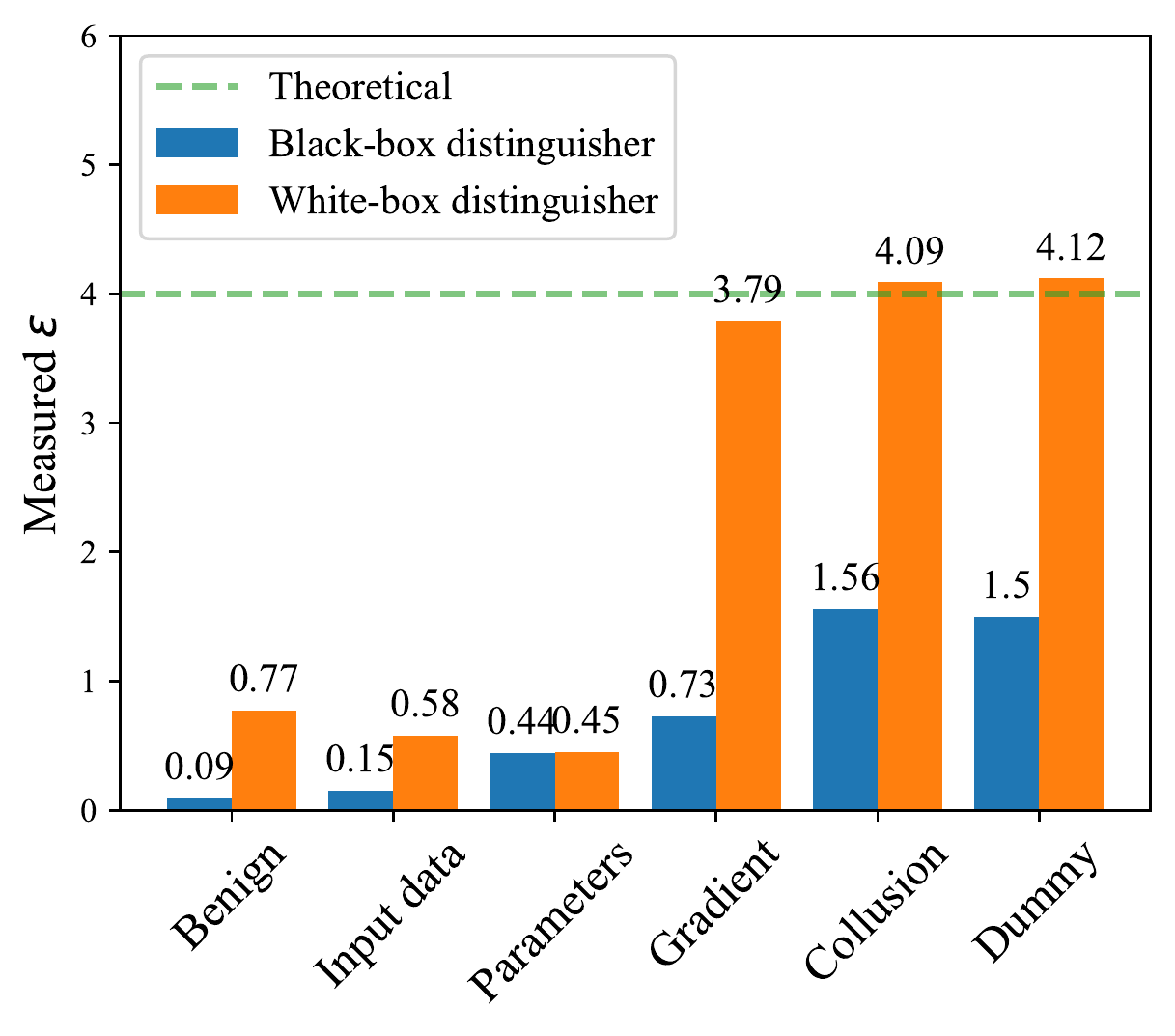}
     \caption{Measured privacy.}
     \label{fig:summary}
     \end{subfigure}
     \caption{Our (a) privacy test demonstrates (b) empirical LDP against six instantiated adversaries. The test consists of \textbf{crafter}, which generates (malicious) inputs, and \textbf{distinguisher}, which infers the input.
     The measured privacy is close to the theoretical bound given by $\epsilon$ when the adversary directly manipulates raw gradients (i.e., \textit{white-box} distinguisher), resulting in far from the bound otherwise.}
\end{figure}

\textbf{Related work.}
\citep{liu2019investigating,236254,auditing-DP,adversary-instantiation} instantiated adversaries in machine learning under CDP.
Liu et al.\cite{liu2019investigating} proposed an interpretation of the CDP by means of a hypothesis test.
An adversary predicts whether the input data is $D$ or $D'$ from the output.
They quantified the success of the hypothesis test using precision-recall-relation.
Jagielski et al.~\citep{auditing-DP} were the first to attempt an empirical privacy measure using a hypothesis test.
They analyzed the privacy level of training models privacy-protected with DP-SGD~\citep{dp-sgd-original,bassily2014private, song2013stochastic} against membership inference~\citep{paper-MI} and two poisoning attacks~\citep{gu2017badnets,auditing-DP}.
Nasr et al.~\citep{adversary-instantiation} showed that theoretical and empirical $\epsilon$ are tight when a contaminated database is used for training in DP-SGD.
As mentioned above, we follow the abstract framework proposed in \citep{adversary-instantiation} but address further challenging issues to realize how to measure the lower bounds of LDP in FL.
ML Privacy Meter \citep{murakonda2020ml} verifies whether a mechanism guarantees privacy, while our study shows how a mechanism that already guarantees $\epsilon$-LDP can withstand attacks.

\section{Preliminaries} \label{sec-prelim}
This section introduces essential knowledge for understanding our
proposal.
We first describe local differential privacy.
Then, we introduce LDP-SGD, a local privacy mechanism in distributed learning.

\subsection{Local Differential Privacy}
\begin{dfn}[$\epsilon$-Local Differential Privacy]
\label{def-ldp}
A randomized algorithm $\mathcal{M}$ satisfies $\epsilon$-local differential privacy, if and only if for any pair of inputs $x, x' \in D$ and for any possible output
$S \in Range(\mathcal{M})$:
\begin{equation}
\text{Pr}[\mathcal{M}(x) \in S] \leq e^\epsilon \cdot \text{Pr}[\mathcal{M}(x') \in S].
\end{equation}
\end{dfn}
Intuitively, we only obtain the indistinguishable information from $\mathcal{M}$ even when the inputs differ.

\begin{figure}[t]
    \begin{minipage}[t]{.49\linewidth}
        \begin{algorithm}[H]
    \caption{LDP-SGD; client-side $\mathcal{A}_{\text{client}}$ \cite{ldp-sgd-google}}
    \label{alg-ldp-sgd}
    \begin{algorithmic}[1]
    \REQUIRE Local privacy parameter: $\epsilon$, current model: $\theta_t \in \mathbb{R}^d$, $\ell_2$-clipping norm: $L$
    \STATE Compute clipped gradient \label{ldp-sgd-clip}\\
    $x \leftarrow \nabla\ell(\theta_t; d)\cdot \min \left \{1, \dfrac{L}{||\nabla\ell(\theta_t;d)||_2}\right\} $
    \STATE
        $z \leftarrow
 \begin{cases}
    L\cdot \frac{x}{||x||_2} &  \textrm{w.p. }  \frac{1}{2} + \frac{||x||_2}{2L}\\
    -L\cdot \frac{x}{||x||_2}& \textrm{otherwise.}
\end{cases}$

    \STATE Sample $v \sim _u S^d$, the unit sphere in $d$ dims \label{ldp-sgd-uniform}\\
    $\hat{z} \leftarrow
        \begin{cases}
            \text{sgn}(\langle z, v \rangle)\cdot v &  \textrm{w.p. } \frac{e^{\epsilon}}{1+e^{\epsilon}} \\
            -\text{sgn}(\langle z, v\rangle)\cdot v & \textrm{otherwise.}
        \end{cases}$ \label{alg:line3}

    \RETURN $\hat{z}$
    \end{algorithmic}
\end{algorithm}
    \end{minipage}
    \hfill
    \begin{minipage}[t]{.49\linewidth}
        \begin{algorithm}[H]
    \caption{LDP-SGD; server-side $\mathcal{A}_{\text{server}}$ \cite{ldp-sgd-google}}
    \label{alg-ldp-sgd-s}
    \begin{algorithmic}[1]
    \REQUIRE Local privacy budget: $\epsilon$, number of epochs: $T$, parameter set: $C$
    \STATE $\theta_0 \leftarrow \{0\}^d$
    \FOR{$t \in [T]$}
    \STATE Send $\theta_t$ to all clients
    \STATE $g_t \leftarrow
    \frac{L\sqrt{\pi}}{2} \cdot \frac{\Gamma(\frac{d-1}{2}+1)}{\Gamma(\frac{d}{2}+1)}\cdot
    \frac{e^{\epsilon}+1}{e^{\epsilon}-1}
    \left(\frac{1}{n}{\sum_{i\in[n]} \hat{z}_i }\right)$
    \STATE Update: ${\theta_{t+1} \leftarrow \prod_C(\theta_t - \eta_t \cdot g_t)}$,\\
    where $\prod_C(\cdot)$ is the $\ell_2$-projection onto set $C$, and $\eta_t = \frac{||C||_2\sqrt{n}}{L\sqrt{d}}\cdot \frac{e^{\epsilon} - 1}{e^{\epsilon} + 1}
    $
    \ENDFOR
    \RETURN $\theta_{priv} \leftarrow \theta_T$
    \end{algorithmic}
\end{algorithm}
    \end{minipage}
\end{figure}

\subsection{Federated Learning using LDP-SGD}
\label{sec:ldp-sgd}

Federated learning (FL) \citep{mcmahan2017communication,kairouz2021advances} is a decentralized machine learning technique.
The major difference between traditional machine learning and FL is that client does not share her data with servers or other clients.
However, in FL, it is pointed out that the images used in training the model can be restored from their gradients released by clients \citep{paper-inverting-grad,paper-cvpr-gradinversion}.
One way to prevent the restoration is to randomize gradients for satisfying LDP.

As shown in Figure. \ref{fig:overview}, FL assumed in this paper consists of an untrusted model trainer (server) and clients that own sensitive data.
First, the client creates the gradient with the parameters distributed by the model trainer.
The client then randomizes the gradient under LDP and sends it to the model trainer.
The model trainer updates the parameters with the gradient collected from the client.
We randomize the gradient using LDP-SGD (locally differentially private stochastic gradient descent) \citep{duchi2018,ldp-sgd-google} as described in Algorithm \ref{alg-ldp-sgd} and \ref{alg-ldp-sgd-s}.
The client side algorithm (Algorithm \ref{alg-ldp-sgd}) performs two gradient randomizations at line 2 and 3.
We refer to the line 2 as \textbf{gradient norm projection} and the line 3 as \textbf{random gradient sampling}.
We instantiate the adversary based on the following two facts:
\begin{itemize}[nosep]
    \item \textbf{Gradient norm projection} (line 2): the smaller the norm of the gradient before randomization, the more likely the sign of the gradient is reversed.
    \item \textbf{Random gradient sampling} (line 3): the gradient close to the gradient before randomization is likely to be generated when the privacy parameter $\epsilon$ is set large.
\end{itemize}



\section{Measuring Lower Bounds of LDP}
This section describes how to measure the empirical privacy level of FL under LDP.

\subsection{Lower Bounding {$\epsilon$} of LDP as Hypothsis Testing}
Given a output $y$ of a randomized mechanism $\mathcal{M}$, consider the following hypothesis testing experiment.
We choose a null hypothesis as input $x$ and alternative hypothesis as $x'$:
\begin{center}
$H_0$: $y$ came from a input $x$\\
$H_1$: $y$ came from a input $x'$
\end{center}
For a choice of a rejection region $S$, the false positive rate(FP), when the null hypothesis is true but rejected, is defined as \Prob{$\mathcal{M}(x) \in S$}.
The false negative rate(FN), when the null hypothesis is false but retained, is defined
as \Prob{$\mathcal{M}(x') \in \bar{S}$}
where $\bar{S}$ is the complement of $S$.
Mechanism $\mathcal{M}$ satisfying $\epsilon$-LDP is equivalent to satisfying the following conditions\citep{paper-kairouz}.

\begin{thm}[Empirical $\epsilon$-Local Differential Privacy]
For any $\epsilon \in \mathbb{R}^+$, a randomized mechanism $\mathcal{M}$ is $\epsilon$-differentially private if and only if the following conditions are satisfied for all pairs of input values $x$ and
$x'$, and all rejection region $S\subseteq Range(\mathcal{M})$:
\begin{equation}
\begin{split}
    \text{Pr}[\mathcal{M}(x) \in S] + e^\epsilon \cdot\text{Pr}[\mathcal{M}(x') \in \bar{S}] &\geq 1, and \\
    e^\epsilon \cdot\text{Pr}[\mathcal{M}(x) \in S] + \text{Pr}[\mathcal{M}(x') \in \bar{S}] &\geq 1.
\end{split}
\end{equation}
\end{thm}
Therefore, we can determine the empirical $\epsilon$-local differential privacy as
\begin{equation}
\label{eq-empirical-eps}
\epsilon_{\text{empirical}} = \max \left ({\log\dfrac{1-\text{FP}}{\text{FN}}, \log\dfrac{1-{\text{FN}}}{\text{FP}}} \right)
\end{equation}

For example, in 100 trials, suppose the false-positive rate, when the actual input was $x$ but the distinguisher guesses $x'$, was 0.1, and the false-negative rate, when the actual input was $x'$ but the distinguisher guesses $x$, was 0.2.
Substitute in FP and FN in Equation \ref{eq-empirical-eps} to obtain $ \epsilon_{\text{empirical}} \fallingdotseq 2.0 $.

\subsection{Instantiating the LDP Adversary in FL}
To perform the privacy test based on the above hypothesis testing, we define the following entities:
\begin{itemize}[nosep]
    \item \textbf{Crafter} produces two gradients, $g_1$ and $g_2$, with the global model $\theta_t$, corresponding to a malicious client in FL. This entity honestly randomizes one of gradients by Algorithm \ref{alg-ldp-sgd} and makes it $\tilde{g}$.
    $\tilde{g}$ is sent to the model trainer and distinguisher.
    \item \textbf{Model trainer} uses the $\tilde{g}$ received from the crafter to update the global model $\theta_t$ to $\theta_{t+1}$ by Algorithm \ref{alg-ldp-sgd-s}. This entity corresponds to a untrusted server in FL.
    \item \textbf{Distinguisher} predicts whether the randomized gradient was $g_1$ or $g_2$. This entity has the data $x_1$, $x_2$ that the crafter used to generate the gradient.
\end{itemize}

We divide the distinguisher into two classes, \textbf{black-box distinguisher} and \textbf{white-box distinguisher}.
For each distinguisher, we propose two different privacy measurement tests,  namely, \textbf{black-box LDP test} (Algorithm \ref{alg-ldp-test-black}) and \textbf{white-box LDP test} (Algorithm \ref{alg-ldp-test-white}).

\begin{figure}[t]
    \begin{minipage}[t]{.49\linewidth}
        \begin{algorithm}[H]
    \caption{Black-box LDP Test in FL}
    \label{alg-ldp-test-black}
    \begin{algorithmic}[1]
    \REQUIRE Privacy parameter: $\epsilon$, $\#$trials: $K$
    \FOR{$k \in [K]$}
    \STATE Model trainer sends $\theta_t$ to crafter.
    \STATE \textbf{Crafter}
    \STATE \hskip1.0em $\{g_1, g_2\} \leftarrow \text{Craft}(x_1, x_2, \theta_t)$
    \STATE \hskip1.0em Randomly choose $g$ from $\{g_1, g_2\}$.
    \STATE \hskip1.0em $\Tilde{g} \leftarrow \mathcal{A}_{\text{client}}(g)$.
    \STATE \hskip1.0em Submit $\Tilde{g}$ to model trainer.
    \STATE \hskip1.0em Share $x_1$, $x_2$, $\theta_t$ w/ distinguisher
    \STATE \textbf{Model Trainer}
    \STATE \hskip1.0em $\theta_{t+1}\leftarrow \mathcal{A}_{\text{server}}(\theta_{t})$
    \STATE \hskip1.0em Submit $\theta_{t+1}$ to distinguisher
    \STATE \textbf{Distinguisher}
    \STATE \hskip1.0em $\text{guess} \leftarrow \mathcal{D}_{\text{black}}(\theta_{t+1},\theta_t,x_1,x_2)$
    \ENDFOR
    \STATE Compute $\epsilon_{\text{empirical}}$ as (\ref{eq-empirical-eps})
    \end{algorithmic}
\end{algorithm}
    \end{minipage}
    \hfill
    \begin{minipage}[t]{.49\linewidth}
        \begin{algorithm}[H]
    \caption{White-box LDP Test in FL}
    \label{alg-ldp-test-white}
    \begin{algorithmic}[1]
    \REQUIRE Privacy parameter: $\epsilon$, $\#$trials: $K$
    \FOR{$k \in [K]$}
    \STATE Model trainer sends $\theta_t$ to crafter.
    \STATE \textbf{Crafter}
    \STATE \hskip1.0em $\{g_1, g_2\} \leftarrow \text{Craft}(x_1, x_2, \theta_t)$
    \STATE \hskip1.0em Randomly choose $g$ from $\{g_1, g_2\}$
    \STATE \hskip1.0em $\Tilde{g} \leftarrow \mathcal{A}_{\text{client}}(g)$.
    \STATE \hskip1.0em Submit $\Tilde{g}$ to distinguisher.
    \STATE \hskip1.0em Share $g_1$, $g_2$ w/ distinguisher.
    \STATE \textbf{Distinguisher}
    \STATE \hskip1.0em $\text{guess} \leftarrow \mathcal{D}_{\text{white}}(\Tilde{g},g_1,g_2)$
    \ENDFOR
    \STATE Compute $\epsilon_{\text{empirical}}$ as (\ref{eq-empirical-eps})
    \end{algorithmic}
\end{algorithm}
    \end{minipage}
\end{figure}

\subsection{Dinstinguisher Algorithms}

\paragraph{Black-box Distinguisher.}
Let black-box distinguisher $\mathcal{D}_{\text{black}}$ be the distinguisher that cannot access the randomized gradient, but can access the global parameters $\theta_{t+1}$ updated with using it.
That is, the black-box disntinguisher only observes the updates after each step of federated learning, but not see any details of inner process of the federated learning.
For simplicity's sake, we only consider the global model is updated using a single randomized gradient.
The black-box distinguisher predicts the randomized gradient by computing the loss of $x_1$ (and $x_2$), then compare the differences before and after the model update.
By default, the distinguisher $\mathcal{D}_{\text{black}}$ employs the following inference against two choices:
\begin{equation}
\label{eq_weak_2}
    guess = \begin{cases}
             {g_1}& f(x_1;\theta_{t+1}) \leq f(x_1;\theta_t)   \\
            {g_2}& otherwise
        \end{cases}.
\end{equation}
This is based on that $g_2$ aims to degrade the model.
However, the inference does not work well for all crafters.
In section \ref{sec-crafter}, we mention about the replacement for the inequality.

\paragraph{White-box Distinguisher.}
White-box distinguisher $\mathcal{D}_{\text{white}}$ directly accesses a single randomized gradient $\tilde{g}$.
In addition, the white-box distinguisher can generate two raw gradients $g_1$ and $g_2$ as well as the client, then computes the similarity between $\tilde{g}$ and the raw gradients to estimate the origin of the $\tilde{g}$.
Here, we employ the cosine similarity between $\tilde{g}$ and $\{g_1, g_2\}$.
The discrimination by the white-box distinguisher is defined as follows:
\begin{equation}
\label{eq_strong_1}
    guess = \begin{cases}
             {g_1}& cos(\tilde{g}, g_1) \geq cos(\tilde{g}, g_2) \\
            {g_2}& otherwise
        \end{cases}.
\end{equation}

\subsection{Crafter Algorithms} \label{sec-crafter}
We propose several types of adversaries depending on the access level.
We introduce six types of crafter algorithms Craft($\cdot$) in Algorithm \ref{alg-ldp-test-black} and \ref{alg-ldp-test-white}.
One crafter algorithm assumes collusion with the model trainer.
The adversary setting proposed in this paper is summarized in Table \ref{table:setting}.

\begin{table}[t]
 \caption{\textbf{Adversary settings.}
 The lower the row, the stronger (i.e. more unrealistic) the adversary.}
 \label{table:setting}
 \centering
 \small
\begin{tabular}{lcccc}
\toprule
\multirow{2}{*}[-2pt]{Adversary}&\multicolumn{2}{c}{Malicious Access}&\multicolumn{2}{c}{Distinguisher} \\\cmidrule(lr){2-3}\cmidrule(lr){4-5}
  & Crafter & Model trainer & Black-box & White-box \\ \midrule
 Benign &Benign. & Benign. & Eq. (\ref{eq_weak_1})&Eq. (\ref{eq_strong_1})\\
 Input perturbation &Images. & Benign. & Eq. (\ref{eq_weak_2})&Eq. (\ref{eq_strong_1})\\
 Parameter retrogression &Parameters. & Benign. & Eq. (\ref{eq_weak_2})&Eq. (\ref{eq_strong_1})\\
Gradient flip & Gradients. & Benign. & Eq. (\ref{eq_weak_2})&Eq. (\ref{eq_strong_1})\\
 Collusion
 &Gradients & Parameters. &Eq. (\ref{eq_weak_2})&Eq. (\ref{eq_strong_1})\\
 Dummy gradient  & Gradients.& Benign.&Eq. (\ref{eq_weak_3})&Eq. (\ref{eq_strong_1})\\
 \bottomrule
\end{tabular}
\end{table}

\textbf{Benign setting.}
As the most realistic setting, suppose all entities are benign.
In the benign setting, the crafter uses the global model $\theta_t$ distributed by the model trainer to generate gradients $g_1$ and $g_2$ from the images $x_1$ and $x_2$.
Thus, the two gradients are also crafted without malicious behaviors as:
\begin{eqnarray*}
        g_1 = \nabla f(x_1, \theta_t); &
        g_2 = \nabla f(x_2, \theta_t)
\end{eqnarray*}
In this scenario, the black-box distinguisher considers the reported randomized gradient as $g_1$ if
\begin{equation} \label{eq_weak_1}
    \delta(x_1) \geq \delta(x_2) \ \text{where} \ \delta(x)=|f(x;\theta_{t+1}) - f(x;\theta_t)| .
\end{equation}
otherwise $g_2$.

\textbf{Input perturbation.} This scenario assumes that the crafter produces malicious input data via perturbations.
The crafter adds perturbation to the data $x_1$ to make it easier to distinguish the gradient.
Although many studies have proposed perturbation algorithms \citep{papernot2016limitations,moosavi2016deepfool,papernot2017practical,carlini2017adversarial,ford2019adversarial},
we adopted the FGSM (Fast Gradient Sign Method) \citep{goodfellow2014explaining}, a well-known method for adversarial perturbations.
For an input image, FGSM uses the gradients of the loss with respect to the input image to create a new image that maximizes the loss.
Let $x_2$ be an adversary pertubed image crafted as $x_2 = x_1 + \alpha \ \text{sgn}(\nabla_{x_1}f(x_1;\theta_t))$ \footnote{We set $\alpha=1$.}.

\textbf{Parameter retrogression.}
Let us assume that the crafter has maliciously manipulated the parameters~\citep{nasr2019comprehensive}.
Generating a gradient from manipulated parameters results in a significant difference in norm compared to generating it from benign parameters.
This setup suggests the ease of distinguishing between two gradients with very different norms.
Let $\theta_t'$ be the parameter updated in the direction of increasing $x_1$ loss as $\theta_t' = \theta_t + \alpha \ \nabla f(x_1;\theta_t)$.
Generate gradients $g_1$ from the benign parameter $\theta_t$ and $g_2$ from the retrograded parameter $\theta_t'$.
\begin{eqnarray*}
    g_1 = \nabla f(x_1, \theta_t); &
    g_2 = \nabla f(x_1, \theta_t')
\end{eqnarray*}

\textbf{Gradient flip.}
In this setting, the crafter processes the raw gradient.
The simplest method of processing a gradient to increase its discriminability is to flip the gradient sign \citep{nasr2019comprehensive}.
As described in Section \ref{sec:ldp-sgd}, in the random gradient sampling of LDP-SGD, when $\epsilon$ is set large, the sign of the gradient is unlikely to change, so we argue that such processing is effective.
The crafter uses the global model $\theta_t$ distributed by the model trainer to generate gradients $g_1$ from the image $x_1$.
Let $g_2$ be a gradient with the sign of $g_1$ flipped.
\begin{eqnarray*}
    g_1 = \nabla f(x_1, \theta_t); & g_2 = -g_1
\end{eqnarray*}

\textbf{Collusion.}
We here consider distributing a malicious model from the server to the clients.
As explained in Section \ref{sec:ldp-sgd}, in the  gradient norm projection of LDP-SGD, the smaller the norm of the raw gradient, the easier it is for the sign to be flipped.
Taking advantage of this property, we consider a setting where the crafter and the model trainer collude to make generating a gradient with a small norm challenging.
With this setting, the model trainer intentionally creates a global model with a massive loss and distributes it to the crafter.
The crafter flips a gradient in the same way as the gradient flipping.
Thus the gradient $g_2$ is generated as $g_2 = -g_1$.
As a realization of this adversary, we introduce the following procedure. First the model trainer generates a malicious global model $\Tilde{\theta}_t$ from only images with the same label and distribute it to the crafter.
Second, the crafter utilizes the malicious model $\Tilde{\theta}_t$ to generate gradients $g_1$ from the image $x_1$.
The label of $x_1$ is different from the label of the images used to generate the malicious model.
\begin{eqnarray*}
    g_1 = \nabla f(x_1, \Tilde{\theta}_t);
     &\textrm{where } \Tilde{\theta}_t \textrm{ is a malicious model.}
\end{eqnarray*}
As well as the gradient flip, the adversary craft $g_2$ by flipping $g_1$.

\textbf{Dummy gradient.}
Here, we consider that the crafter produces a \textit{dummy} gradient.
LDP-SGD includes the gradient norm projection, which causes errors due to randomly flipping gradients.
The intuition here is to craft a dummy gradient that never cause such errors.
The simplest way to do this is to generate a gradient with a large norm, regardless of the image or model the crafter has.
%
To craft the dummy gradient, we simply fill a constant $\lambda$ into all the elements of the gradient as:
\begin{eqnarray*}
    g_1 = (\lambda, \lambda, \dots, \lambda);
\end{eqnarray*}
 The norm of $g_1$ must be greater than or equal to the clipping threshold $L$ to avoid the gradient norm projection.
Therefore, let $\lambda$ be $L/\sqrt{d}$ and $d$ be the dimension of the gradient.
For another choice of the gradient, we also employ the gradient flipping to maximize the difference against the dummy gradient.
Thus the gradient $g_2$ is generated as $g_2 = -g_1$ as well.
In this scenario, the black-box distinguisher considers the reported randomized gradient as $g_1$ if
\begin{equation}
    \sum_{i\in [d]} \text{sgn}(\theta^i_{t+1} - \theta^i_t) \geq 0,
    \label{eq_weak_3}
\end{equation}
otherwise $g_2$. This is due to that the crafter produces the dummy gradient $g_1$ by filling the positive constant.
The combination of crafting the dummy gradient and flipping it is a worst-case attack for the LDP adversary.
We give a proof below.


\subsection{Analysis of the Worst-case Attack}


In the above adversary instatiations with the crafter and the distinguisher, we can find a worst-case attack in FL using LDP-SGD.
About the worst-case, we meet the following proposition.
\begin{proposition}
The white-box distinguisher with the crafter producing a raw gradient whose $\ell_2$-norm is larger than or equal to $L$ is a worst-case attack reaching the theoretical upper bound given by $\epsilon$.
\end{proposition}

\begin{proof}
Recall that in LDP-SGD, the direction is reversed depending on the gradient norm and $\epsilon$.
First, since the sign is easily flipped when the norm of the gradient is small in the gradient norm projection, the norm of the gradient must be greater than or equal to $L$.
Hereafter, we assume that $g$ is a gradient with norm $L$.
Let $g_{\text{flip}}$ be the flipped gradient of $g$, $\tilde{g}$ and $\tilde{g}_{\text{flip}}$ be the outputs of the random gradient sampling, respectively.
With probability of $\frac{e^\epsilon}{1+e^\epsilon}$, the following equation holds on line \ref{alg:line3}.
\begin{eqnarray*}
    g\cdot\tilde{g} \geq 0 &\quad g_{\text{flip}} \cdot \tilde{g}_{\text{flip}} \geq 0
\end{eqnarray*}
Let $\tilde{g}_{\text{target}}$ be the gradient randomized of either $g$ or $g_{\text{flip}}$ with probability of $50\%$.
Let us predict what $g_{\text{target}}$ be.
The case can be divided by the cosine similarity between $\tilde{g}_{\text{target}}$ and $g$ as follows:
\begin{eqnarray*}
cos(\tilde{g}_{\text{target}}, g) \geq 0
&\Rightarrow&
 \begin{cases}
            g_{\text{target}} \textrm{ was } g \textrm{ and } g \textrm{ rotation was within 90 degrees.} &  \textrm{ w.p. } \frac{1}{2}\cdot\frac{e^\epsilon}{1+e^\epsilon} \\
            g_{\text{target}} \textrm{ was } g_{\text{flip}} \textrm{ and } g_{\text{flip}} \textrm{ rotation was over 90 degrees.\;\;}& \textrm{ w.p. } \frac{1}{2}\cdot\frac{1}{1+e^\epsilon}
        \end{cases}
\\
cos(\tilde{g}_{\text{target}}, g) \leq 0 &\Rightarrow&
 \begin{cases}
            g_{\text{target}} \textrm{ was } g \textrm{ and } g \textrm{ rotation was over 90 degrees.} &  \textrm{w.p. } \frac{1}{2}\cdot\frac{1}{1+e^\epsilon} \\
            g_{\text{target}} \textrm{ was } g_{\text{flip}} \textrm{ and } g_{\text{flip}} \textrm{ rotation was within 90 degrees.}& \textrm{w.p. } \frac{1}{2}\cdot\frac{e^\epsilon}{1+e^\epsilon}
        \end{cases}
\end{eqnarray*}
White-box distinguisher predicts $g$ for the original data of $\tilde{g}_{\text{target}}$ when the cosine similarity between $g$ and $\tilde{g}_{\text{target}}$ is positive, then the prediction is correct with probability $\frac{1}{2}\cdot\frac{e^\epsilon}{1+e^\epsilon}$.
Likewise, white-box distinguisher predicts $g_{\text{flip}}$ for the original data of $\tilde{g}_{\text{target}}$ when the cosine similarity between $g$ and $\tilde{g}_{\text{target}}$ is negative, then the prediction is correct with probability $\frac{1}{2}\cdot\frac{e^\epsilon}{1+e^\epsilon}$.
Therefore, white-box distinguisher can distinguish $g$ and $g_{\text{flip}}$ with probability $\frac{e^\epsilon}{1+e^\epsilon}$.
\end{proof}

\section{Numerical Observations} \label{sec-exp}
Here, we address an experimental study to observe numerical results of our LDP tests in FL.

\subsection{Experimental Setup}
For each adversary instantiations listed in Table \ref{table:setting}, we run ten measurements and average over these measurement results.
Each measurement consists of $K$=10,000 trials by each distinguisher.


\textbf{Hyper-parameters.}
We have implemented our algorithms in Pytorch. 
For each dataset, we used cross-entropy as loss function.
We used three layers convolutional neural network for all datasets.
To train the models, we use LDP-SGD \citep{ldp-sgd-google} with clipping norm size $L$=1.
Guaranteed privacy level $\epsilon$ is set to 0.5, 1.0, 2.0 and 4.0.

\textbf{Datasets.}
We perform experimental evaluations on four datasets, namely MNIST \citep{lecun-mnisthandwrittendigit-2010}, CIFAR-10 \citep{alex-cifar-10}, Fashion-MNIST \citep{xiao2017fashion} and SVHN \citep{netzer2011reading}.
Detail results of CIFAR-10, Fashion-MNIST and SVHN are in Appendix.

 \begin{figure}[t]
  \centering
  \begin{subfigure}{.3\linewidth}
    \centering
    \includegraphics[width=\linewidth]{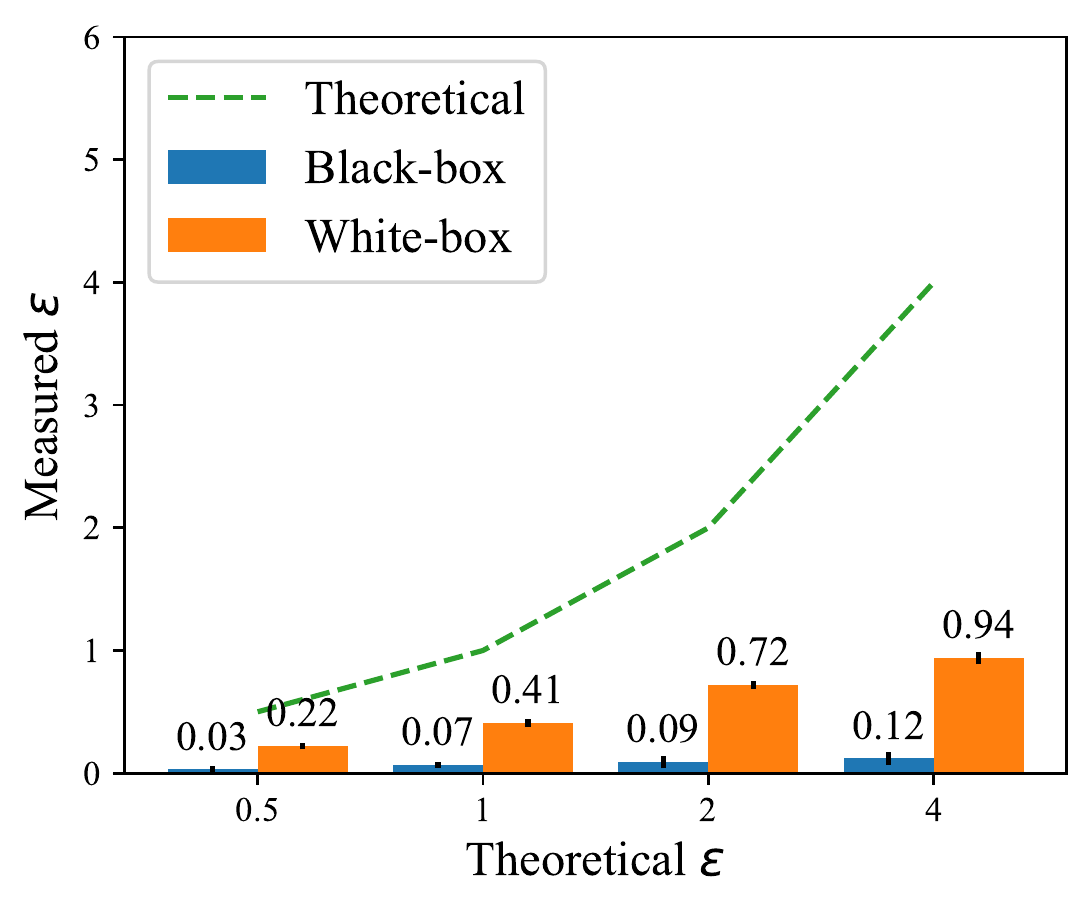}
     \caption{Benign setting}
     \label{fig:benign_setting_mnist}
  \end{subfigure}
  \hfill
  \begin{subfigure}{.3\linewidth}
    \centering
    \includegraphics[width=\linewidth]{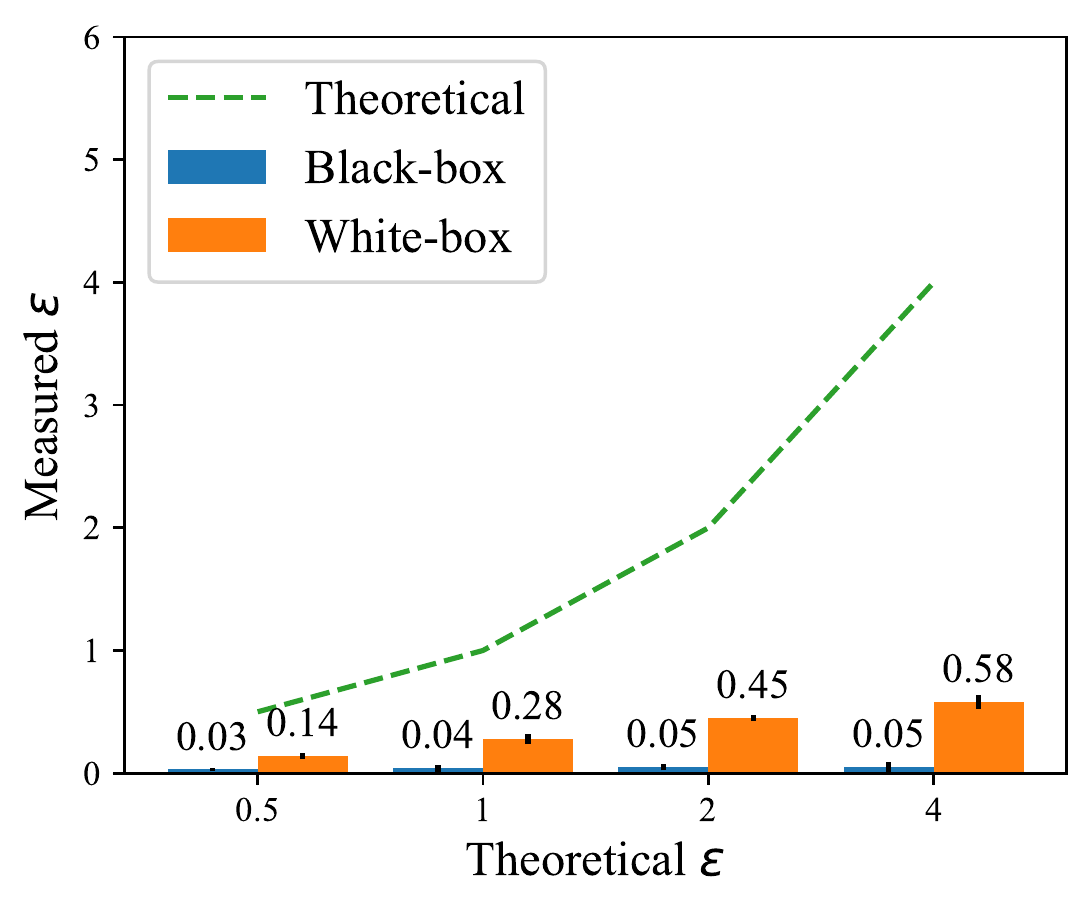}
     \caption{Image perturbation}
     \label{fig:image_perturbation_mnist}
  \end{subfigure}
  \hfill
  \begin{subfigure}{.3\linewidth}
    \centering
    \includegraphics[width=\linewidth]{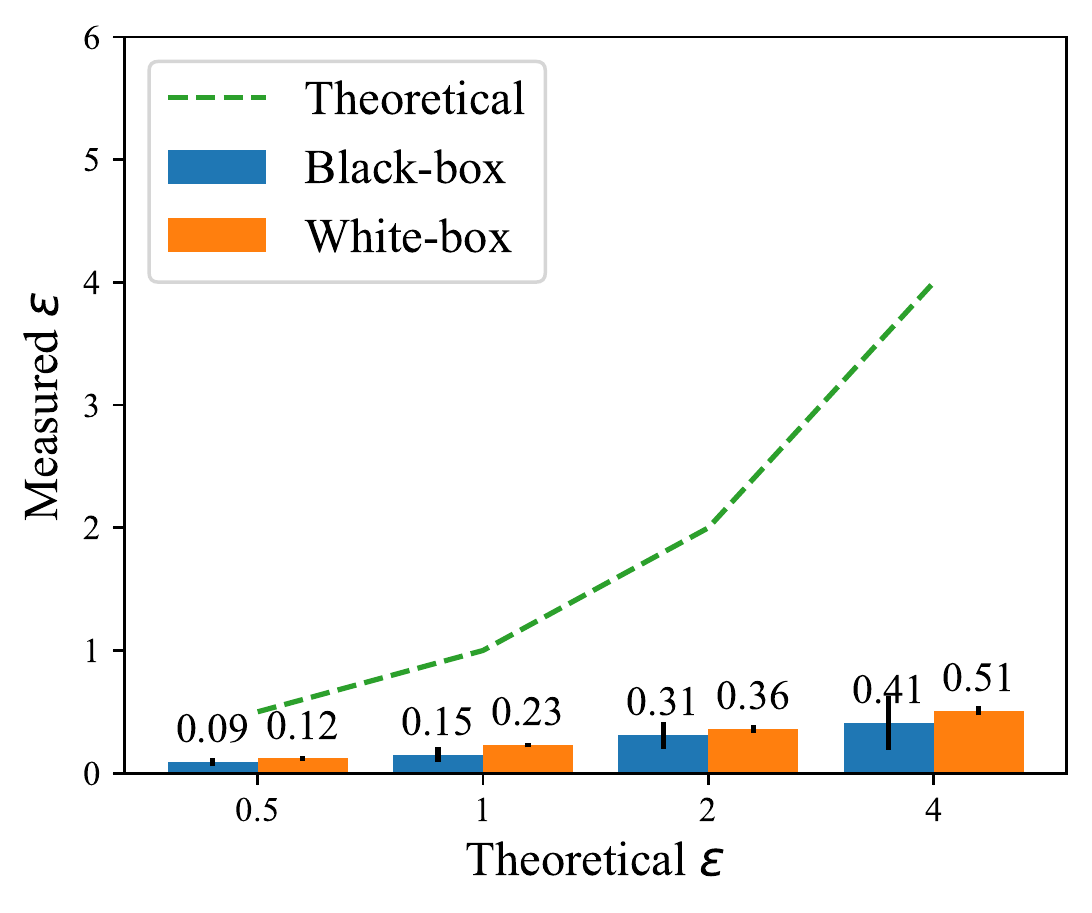}
     \caption{Parameter retrogression}
     \label{fig:param_retrogression_mnist}
  \end{subfigure}

  \begin{subfigure}{.3\linewidth}
    \centering
    \includegraphics[width=\linewidth]{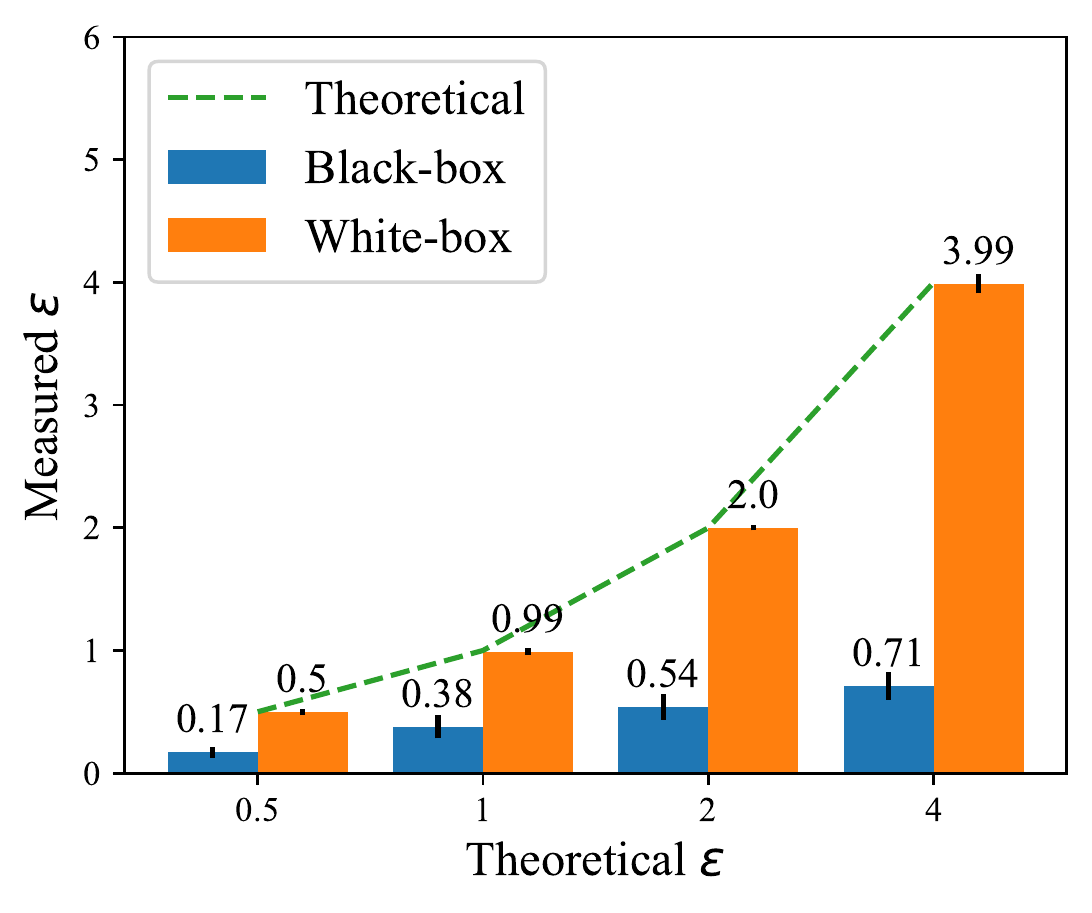}
     \caption{Gradient flip}
     \label{fig:gradflip_mnist}
  \end{subfigure}
  \hfill
  \begin{subfigure}{.3\linewidth}
    \centering
    \includegraphics[width=\linewidth]{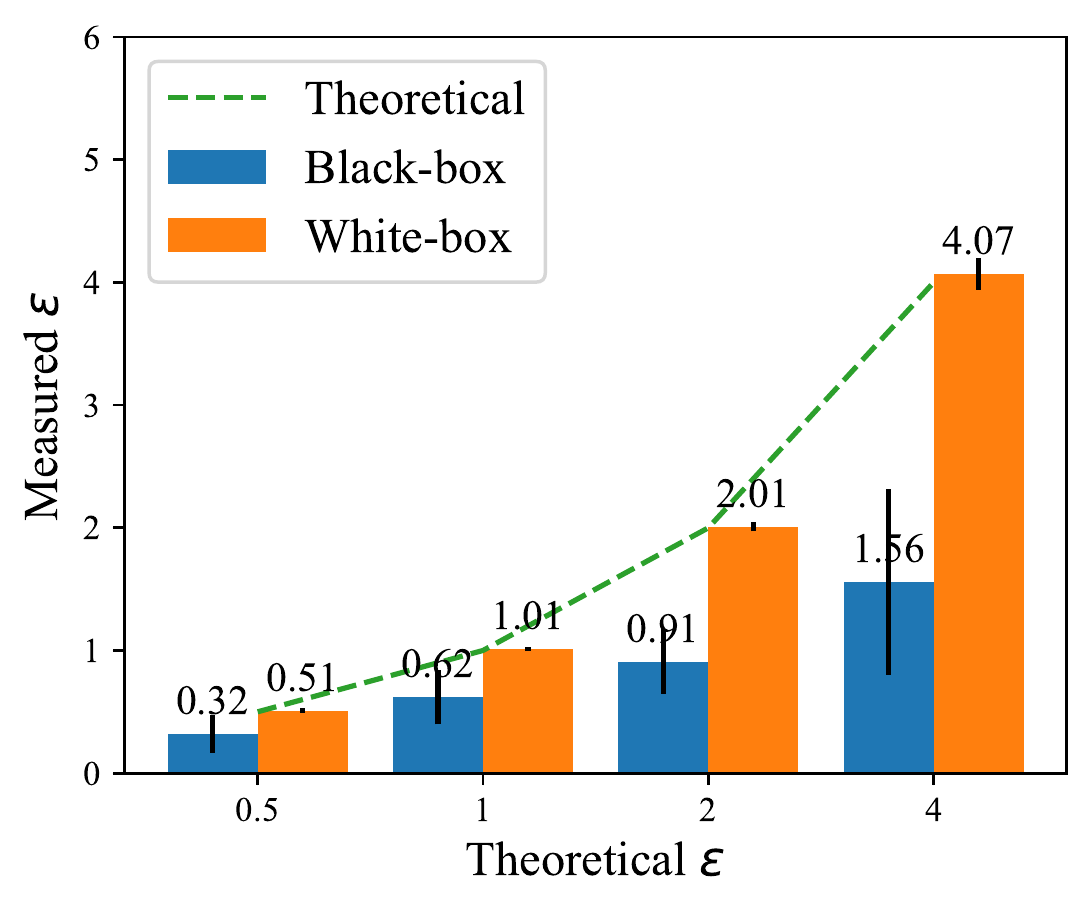}
     \caption{Collusion}
     \label{fig:gradflip_mm_mnist}
  \end{subfigure}
  \hfill
  \begin{subfigure}{.3\linewidth}
    \centering
    \includegraphics[width=\linewidth]{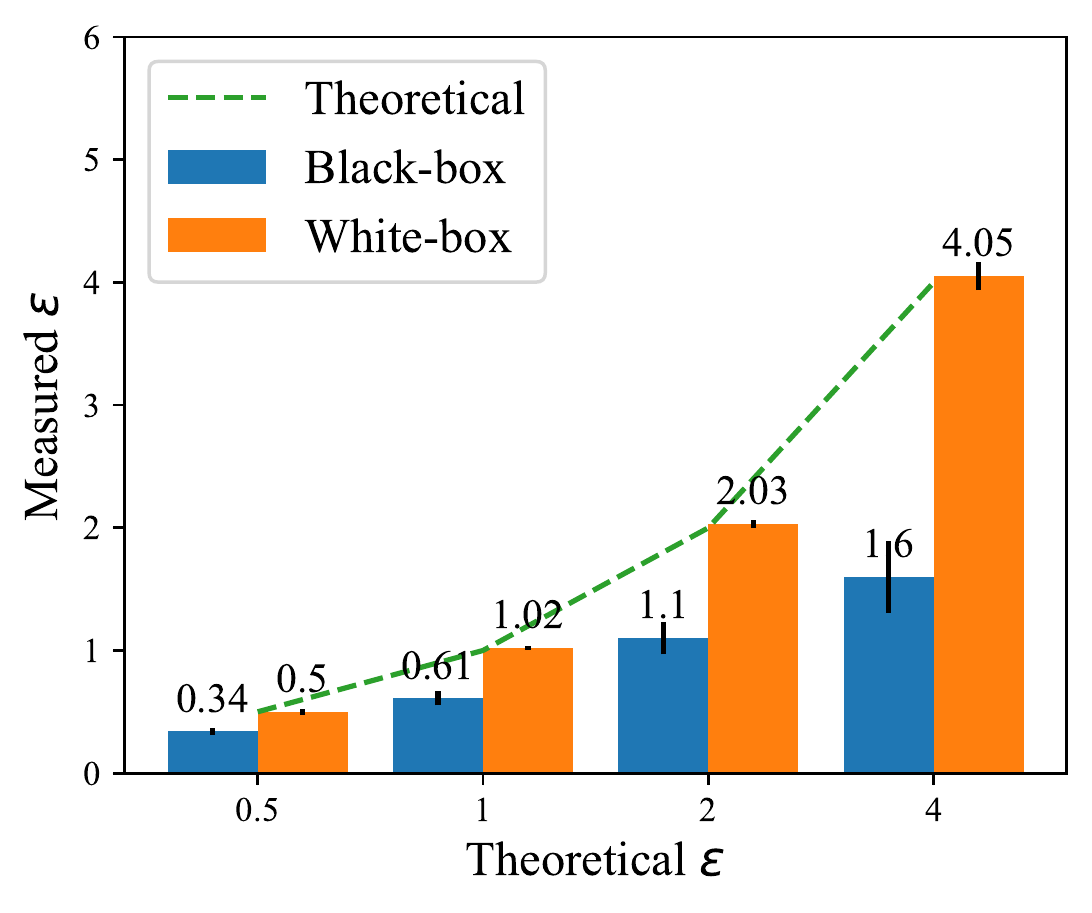}
     \caption{{Dummy gradient}}
     \label{fig:gradflip_maximize_mnist}
  \end{subfigure}

  \caption{The empirical privacy in federated learning. (MNIST)}
  \label{fig:result-mnist}
\end{figure}

\subsection{Observations of Empirical LDP}
\label{sec:observation}
Figure \ref{fig:result-mnist} shows the measured $\epsilon$ ($\epsilon_{\text{empirical}}$) by our LDP test varying the privacy parameter $\epsilon$.

\textbf{Benign.}
In this most realistic setting, even if we set $\epsilon=4$, the empirical privacy strength is $0.94$, and the gap between the theoretical and the empirical is large. (Figure \ref{fig:benign_setting_mnist})
$\epsilon_{\text{empirical}}=0.94$ is equivalent to about $68.1\%$ in terms of the identifiable probability.
White-box distinguisher has a higher $\epsilon_{\text{empirical}}$ than black-box distinguisher.
Therefore it is easier to determine the randomized gradient $\tilde{g}$ itself than the parameter $\theta_{t+1}$ updated with the randomized gradient $\tilde{g}$.

\textbf{Input perturbation.}
From Figure \ref{fig:image_perturbation_mnist}, the probability of discrimination does not increase even if the gradient of the perturbed image is randomized.
The gap between empirical and theoretical privacy strength in this setting is more expansive than in the benign setting.
Therefore, even if the image changes slightly due to perturbation, the gradient will not be easy to distinguish.

\textbf{Parameter retrogression.}
In this setup, we use a nonmalicious gradient $g_1$ and a $g_2$ generated from parameters that have been processed to increase loss.
Since the norm of $g_2$ is larger than the norm of $g_1$, we assume that discrimination would be easier if it were not randomized.
However, Figure \ref{fig:param_retrogression_mnist} shows that the gradient norm has a negligible effect on the discrimination probability.
This result is that in LDP-SGD, gradient clipping keeps large normed gradients below a threshold.

\textbf{Gradient flip.}
The measured $\epsilon$ is shown in Figure \ref{fig:gradflip_mnist}.
Discriminability was improved compared to the benign setting when the crafter flipped the gradient sign.
In particular, when $\epsilon=4$, $\epsilon_{\text{empirical}}$ shows 3.99, almost a theoretical value in white-box distinguisher.
As shown in the worst-case proof, LDP-SGD keeps the direction roughly proportional to the guaranteed $\epsilon$, so $g_1$ and $g_2$, which reverses the sign of $g_1$, are easy to distinguish.

\textbf{Collusion.}
The experimental results in Figure \ref{fig:gradflip_mm_mnist} show that the collusion between the crafter and the model trainer narrows the gap between theoretical privacy and empirical privacy compared to the gradient flipping.
In white-box distinguisher of this setting, the empirical $\epsilon$ reached the theoretical value at all the setting values $\epsilon = 0.5, 1, 2, 4$.
$\epsilon_{\text{empirical}}=4.07$ has the same meaning as successful discrimination with a probability of about $98.2\%$.
In the black-box setting, this collusion increased the empirical LDP against the weaker settings described in the above paragraphs.

\textbf{Dummy gradient.}
From Figure \ref{fig:gradflip_maximize_mnist}, as with the collusion, this setting increases the probability of gradient discrimination.
Figure \ref{fig:result-grad-norm} shows the effect of the norm on the discrimination probability: even with a dummy gradient, if the norm is smaller than the clipping size $L$, the gap between the theoretical and measured $\epsilon$ becomes larger.
Daring to generate the dummy gradients with large norm is not common in FL.
However, it is clear that maximizing the norm of the gradient and sign flipping is the strongest attack.

\textbf{Summary of Results.}
Figure \ref{fig:summary}, \ref{fig:summary_black}, and \ref{fig:summary_white} summarizes the results for each dataset.
Over the four datasets, adversaries who directly manipulate raw gradients can reach the theoretical upper bounds given by the privacy parameter $\epsilon$.
Even in white-box scenarios, adversaries who only access the input data and parameters but not the raw gradient have limited capability.
Moreover, adversaries who only refer to the updated model but not any randomized gradient (i.e., black-box scenarios) have the limited capability,
collusion with the server employing the malicious model, increases the adversary's capability in the black-box setting.

\begin{figure}[t]
\begin{minipage}{.35\linewidth}
    \centering
    \includegraphics[width=.95\linewidth]{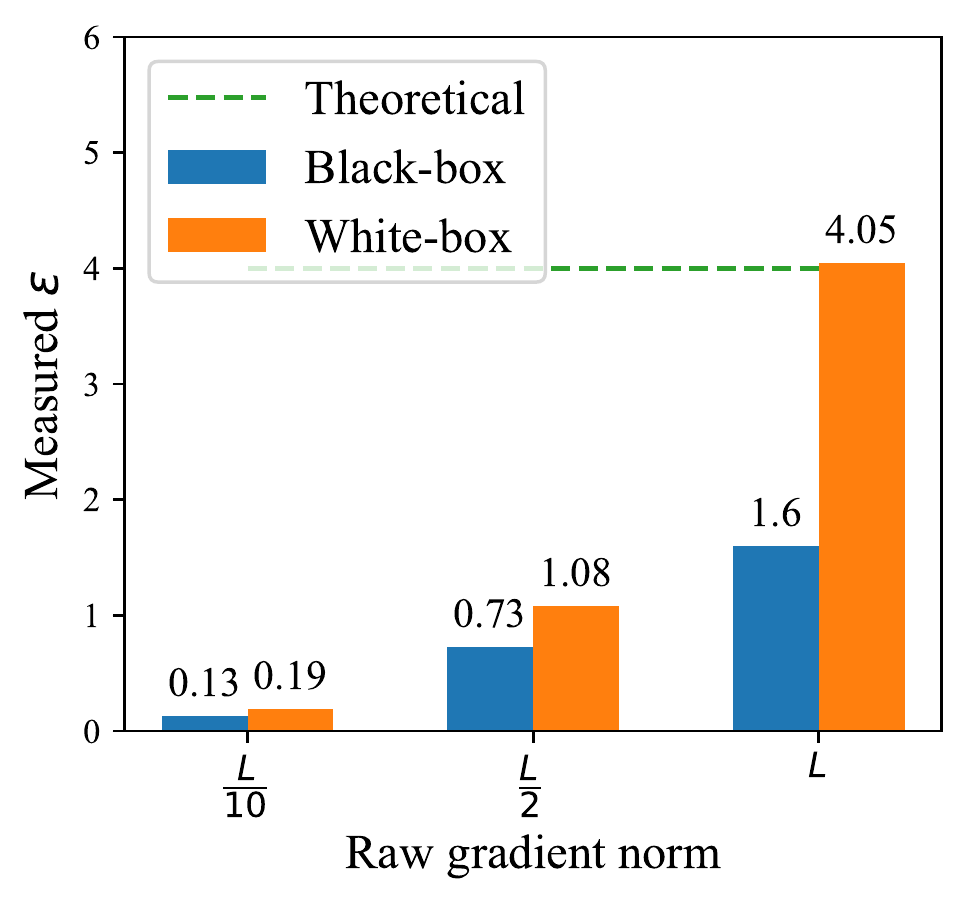}
     \caption{Effect of gradient norm.}
     \label{fig:result-grad-norm}
\end{minipage}
\hfill
\begin{minipage}{.63\linewidth}
    \centering
    \begin{subfigure}{.48\linewidth}
        \centering
        \includegraphics[width=\linewidth]{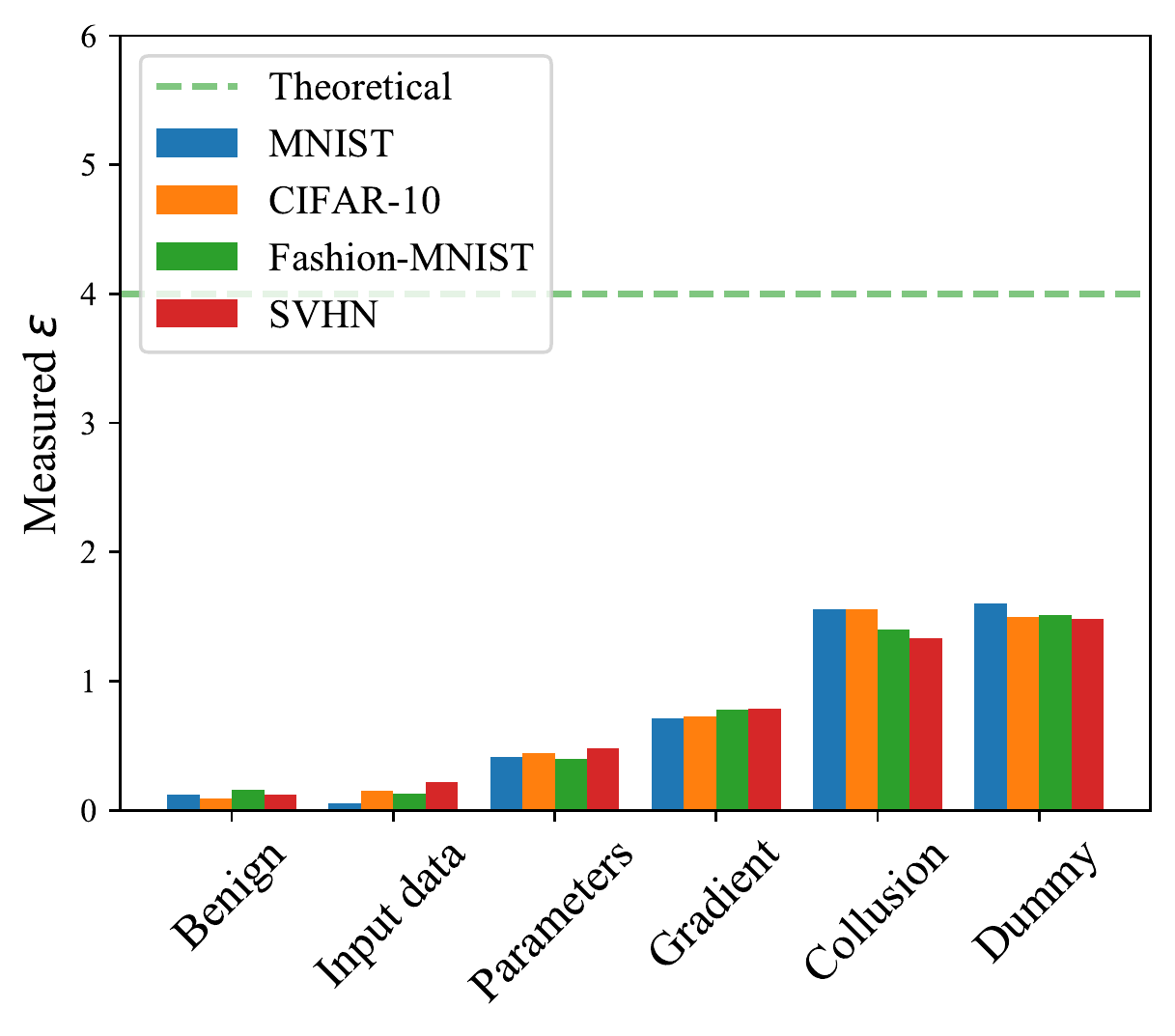}
         \caption{Black-box}
         \label{fig:summary_black}
    \end{subfigure}
    \begin{subfigure}{.48\linewidth}
        \centering
        \includegraphics[width=\linewidth]{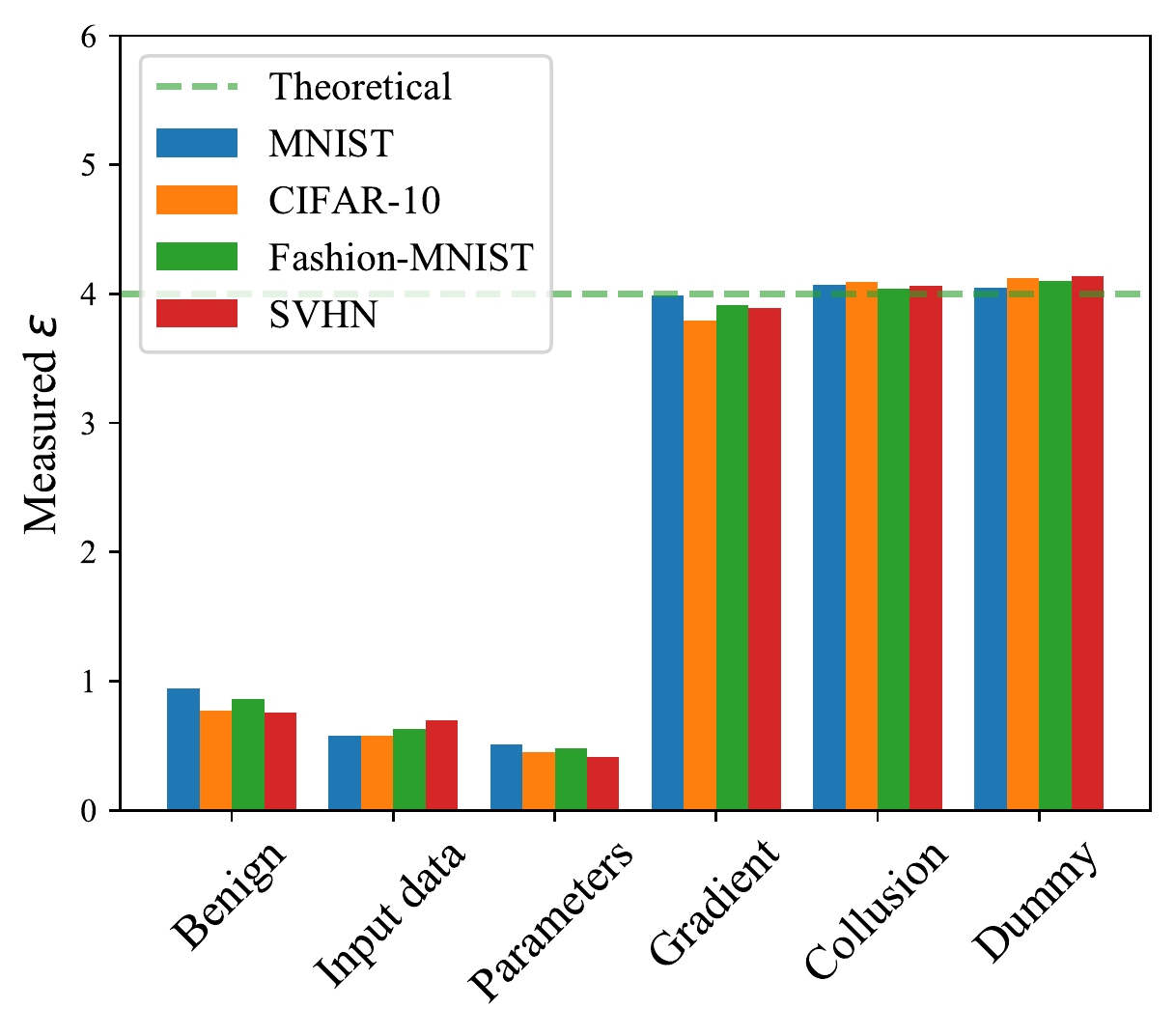}
         \caption{White-box}
         \label{fig:summary_white}
      \end{subfigure}
    \caption{Summary of results in four datasets.}
    \label{fig:result-summary}
\end{minipage}
\end{figure}

\subsection{Discussion towards Relaxations of Privacy Parameters}
The white-box test resulted in tightly close to the theoretical bound defined by the privacy parameters $\epsilon$.
While the black-box test demonstrated stronger privacy levels that were far from the bound.
To prevent the white-box attacks, we can install a secure aggregator based on such as multi-party computations and trusted execution environments (TEEs) \citep{mckeen2013innovative,costan2016intel}.
Using such secure aggregators for updating the global model, the process from reporting gradients to updating the global model is fully encrypted.
The experimental results above do not allow any relaxation only with local randomizers, but with shuffler~\citep{bittau2017prochlo,erlingsson2019amplification, feldman2022hiding, liew2022network} or aggregator implemented by such secure computations~\citep{bonawitz2017practical, kato2022olive}, relaxation of $\epsilon$ can be reasonable in the view of empirical privacy risks.
Even when the shuffler is untrusted and the gradient is randomized with a relaxed privacy parameter (like $\epsilon=8$), the empirical privacy risk against the shuffler might not be significant (like $\epsilon_{\text{empirical}}=1$).
Another possible way is to prohibit gradient poisoning.
As shown in the experiments, without directly manipulating gradients, the empirical privacy is significantly far from the given privacy bound.
If the range of possible input training data is limited, we can introduce out-of-distribution detectors before randomizing raw gradients.
In other words, by adding some restrictions to the crafter, the model trainer, and the distinguisher, it may also be possible to relax $\epsilon$.

\section{Conclusion}
We introduced the empirical privacy test by measuring the lower bounds of LDP.
We then instantiated six adversaries in FL under LDP to measure empirical LDP at various attack surfaces, including the worst-case attack that reached the theoretical upper bound of LDP.
We also demonstrated numerical observations of the measured privacy in these adversarial settings.
In the end, we discussed the possible relaxation of privacy parameter $\epsilon$ with using a trusted entity.

\newpage

\appendix
\section{Experimental Details}
This appendix provides details of the setup and additional experimental results.

\subsection{Experimental Setup}
\textbf{Datasets.}
We run experiments on four datasets:
\begin{itemize}
    \item MNIST: We use the MNIST image dataset, which consists of 28$\times$28 handwritten grayscale digit images, and the task is to recognize the 10 class digits. 
    \item CIFAR-10: We use CIFAR10, a standard benchmark dataset consisting of 32$\times$32 RGB images. The learning task is to classify the images into ten classes of different objects.
    \item Fashion-MNIST: Each example of Fashion-MNIST is a 28$\times$28 grayscale image associated with a label from 10 classes.
    \item SVHN: This dataset is obtained from house numbers in Google Street View images.
SVHN consists of 32$\times$32 RGB images, and the task is to recognize the 10 class digits.
\end{itemize}

\textbf{Neural Network Architecture.}
The neural network we use in the experiments is as in Table \ref{tab:nn-arch}.

\begin{table}[h]
\begin{minipage}{0.54\linewidth}
\begin{center}

\caption{Neural network architecture in measurements.}
    \label{tab:nn-arch}
\begin{tabular}{c|c}\toprule
        Layer & Parameters\\ \midrule
        Convolution & 16 filters of 8$\times$8, strides 2 \\
        Max-Pooling & 2$\times$2 window, strides 2\\
        Convolution &32 filters of 4$\times$4, strides 2 \\
        Max-Pooling & 2$\times$2 window, strides 2\\
        Linear & 32 units\\
        Softmax & 10 units \\\bottomrule
    \end{tabular}
\end{center}
\end{minipage}
\hfill
\begin{minipage}{0.4\linewidth}
\begin{center}
\caption{The black-box test for multiple clients in the benign crafter. (MNIST)}
    \label{tab:num_clients}
\begin{tabular}{c|c}\toprule
        Number of clients & Measured $\epsilon$\\ \midrule
        1 & 0.12 \\
        2 & 0.07\\
        4 & 0.06 \\
        10 & 0.06\\\bottomrule
    \end{tabular}
\end{center}
\end{minipage}
\end{table}

\subsection{Additional Observations}

\textbf{Dataset Perspectives.}
Figure \ref{fig:result-cifar10}, \ref{fig:result-fmnist} and \ref{fig:result-svhn} show that the empirical privacy in federated learning in CIFAR-10, Fashion-MNIST, and SVHN datasets.
The observation of $\epsilon_{\text{empirical}}$ using four datasets shows that $\epsilon_{\text{empirical}}$ is data-independent and has a similar trend.
As in Section \ref{sec:observation}, the white-box LDP test is more potent than the black-box LDP test.
Furthermore, the craft of dummy gradients is closest to the theoretical $\epsilon$. 

\textbf{White-box test showed higher $\epsilon_{\text{empirical}}$ than black-box test.}
This is a natural result in black-box scenarios since the model trainer aggregate gradients such as averaging, learning rate multiplication, and $\ell_2$-projection.
The more clients used to update parameters in FL, the more pronounced this trend becomes.
Table \ref{tab:num_clients} shows that the black-box test for multiple clients becomes more difficult as the number of clients increases.

\textbf{Input perturbation and parameter retrogression were not successful.}
As mentioned in Section \ref{sec:ldp-sgd}, LDP-SGD includes the clipping gradient and the random gradient sampling.
The clipping gradient ignores large differences between the norms of $g_1$ and $g_2$.
Furthermore, the random gradient sampling rotates the gradient by at most 90 degrees, so any slight difference in the direction of $g_1$ and $g_2$ is ignored.



\begin{figure}[t]
  \centering
  \begin{subfigure}{.29\linewidth}
    \centering
    \includegraphics[width=\linewidth]{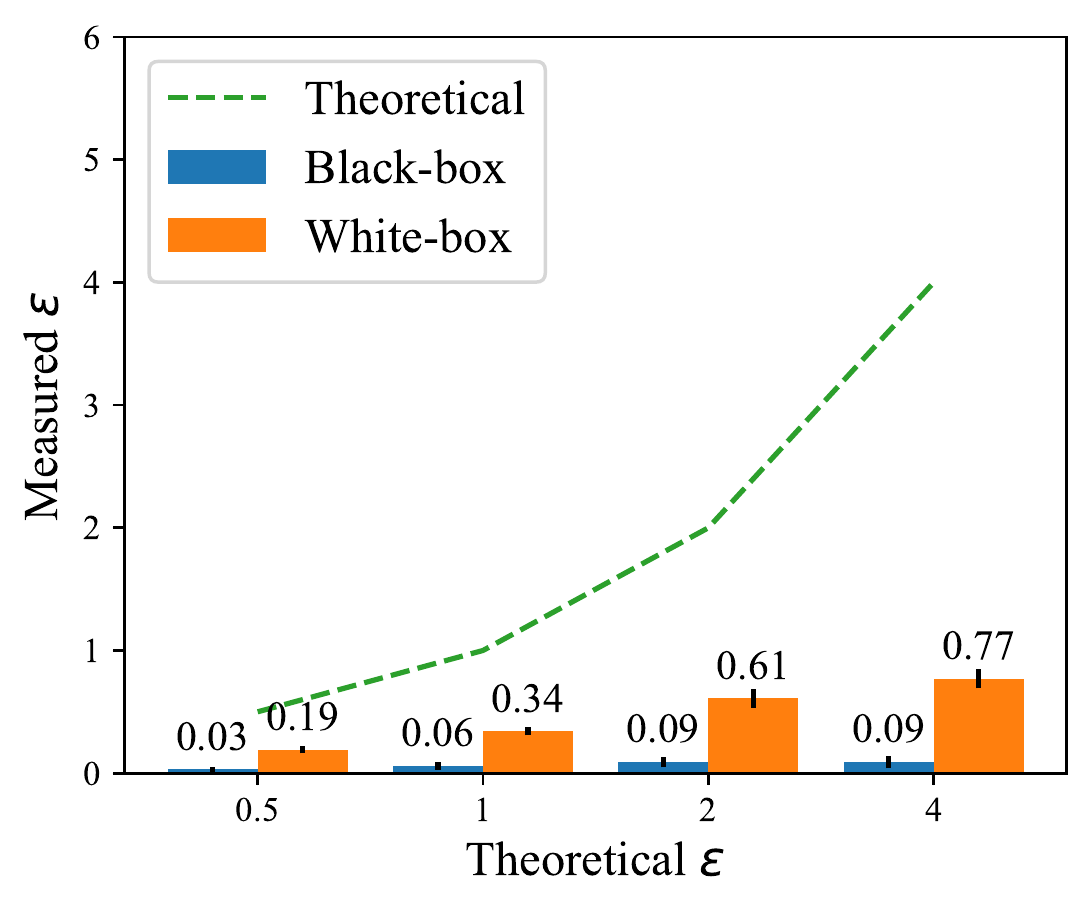}
     \caption{Benign setting}
     \label{fig:benign_setting_cifar}
  \end{subfigure}
  \begin{subfigure}{.29\linewidth}
    \centering
    \includegraphics[width=\linewidth]{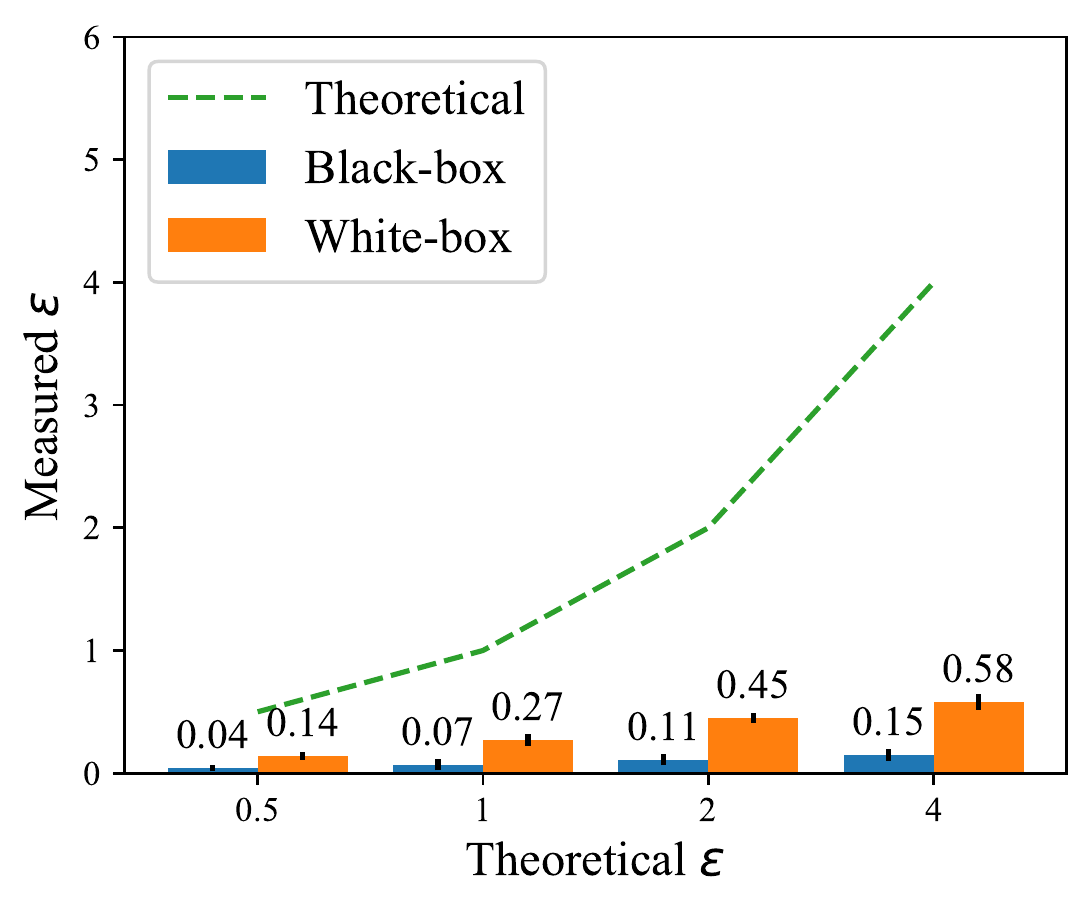}
     \caption{Image perturbation}
     \label{fig:result_perturbation_cifar}
  \end{subfigure}
  \begin{subfigure}{.29\linewidth}
    \centering
    \includegraphics[width=\linewidth]{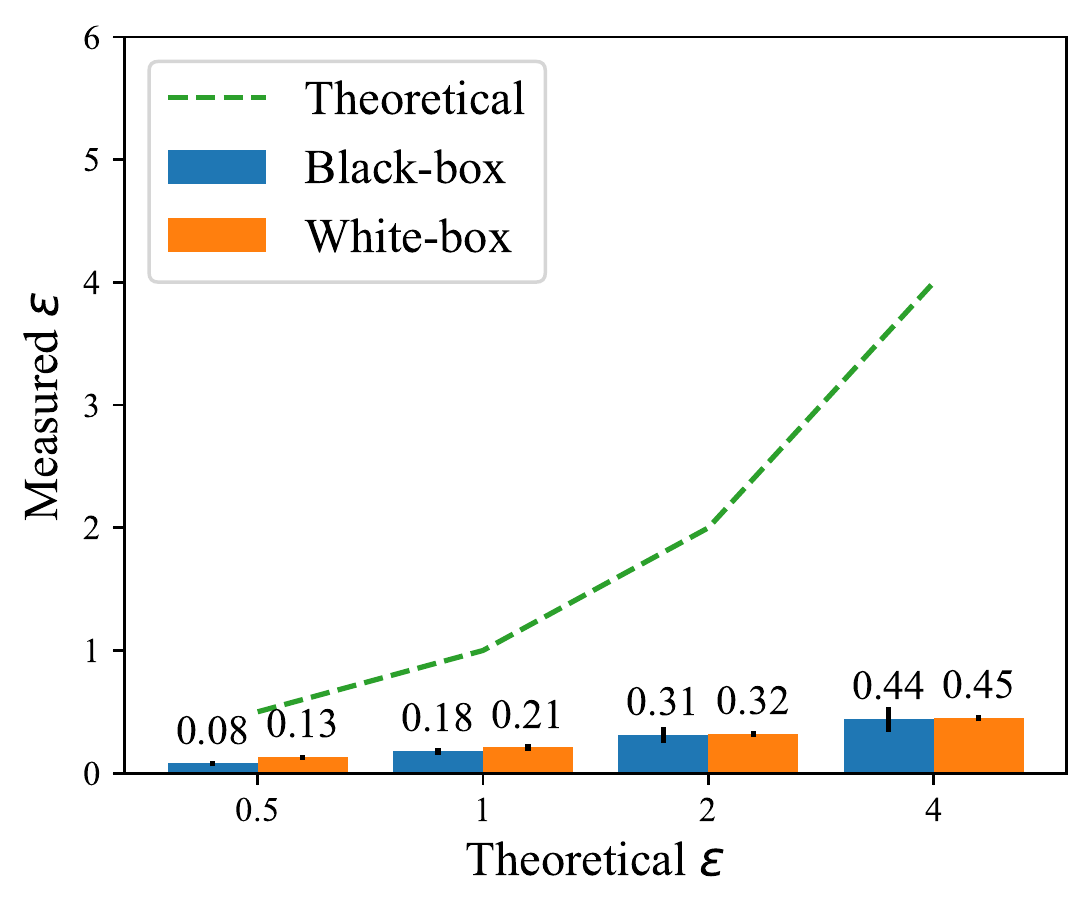}
     \caption{Parameter retrogression}
     \label{fig:result_parameter_cifar}
  \end{subfigure}
  
  \begin{subfigure}{.29\linewidth}
    \centering
    \includegraphics[width=\linewidth]{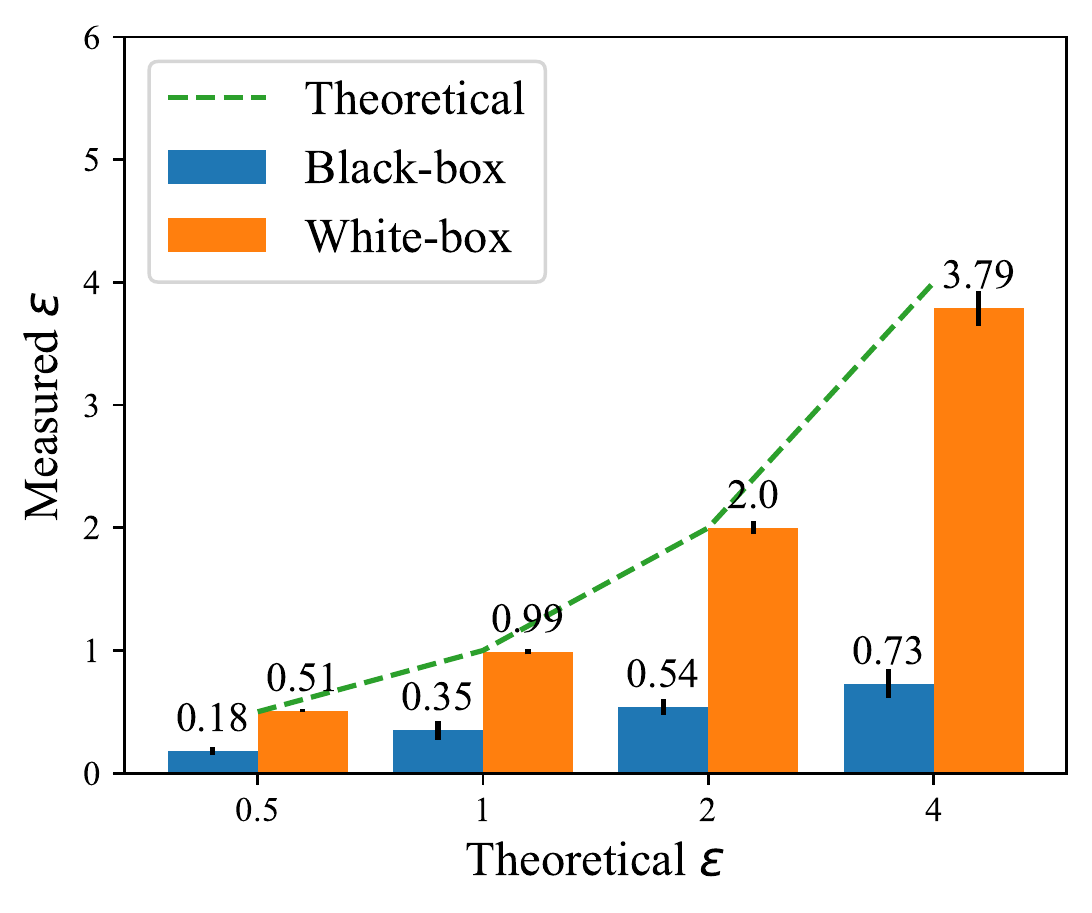}
     \caption{Gradient flip}
     \label{fig:result_gradflip_cifar}
  \end{subfigure}
  \begin{subfigure}{.29\linewidth}
    \centering
    \includegraphics[width=\linewidth]{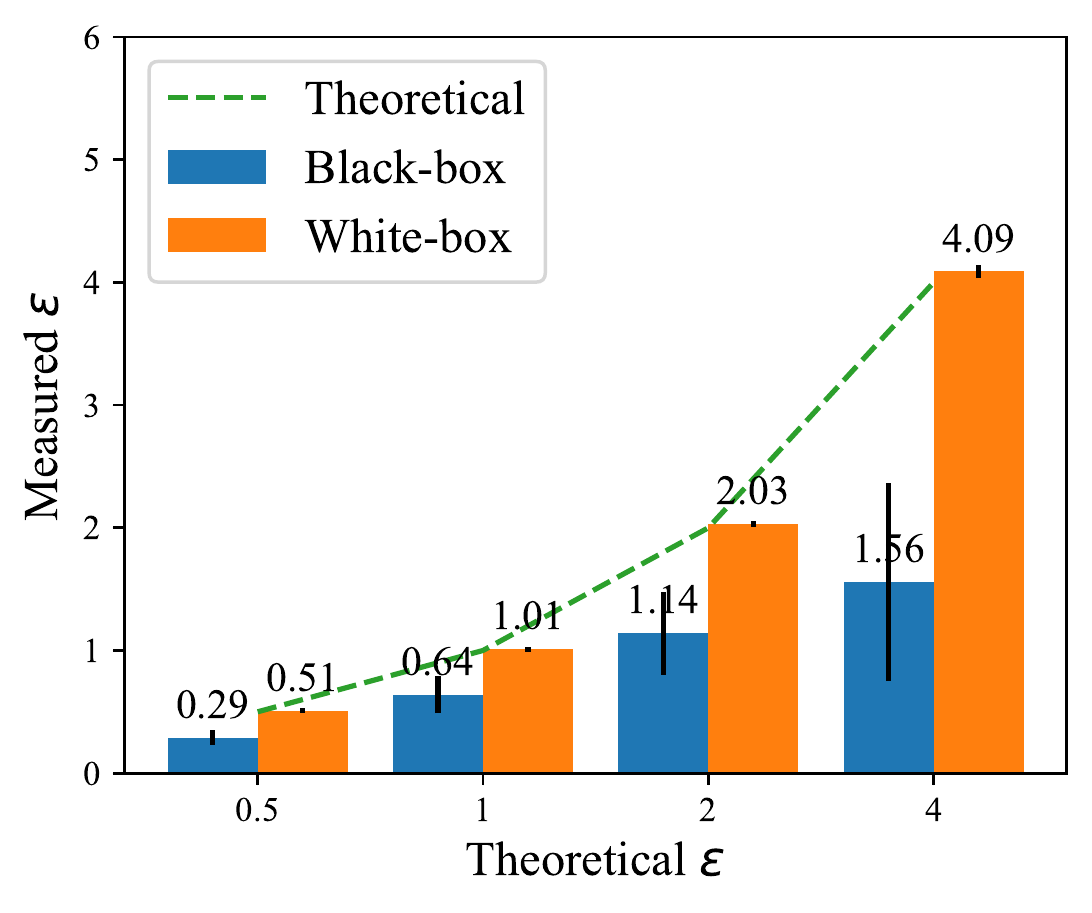}
     \caption{Collusion}
     \label{fig:result_graflip_mm_cifar}
  \end{subfigure}
  \begin{subfigure}{.29\linewidth}
    \centering
    \includegraphics[width=\linewidth]{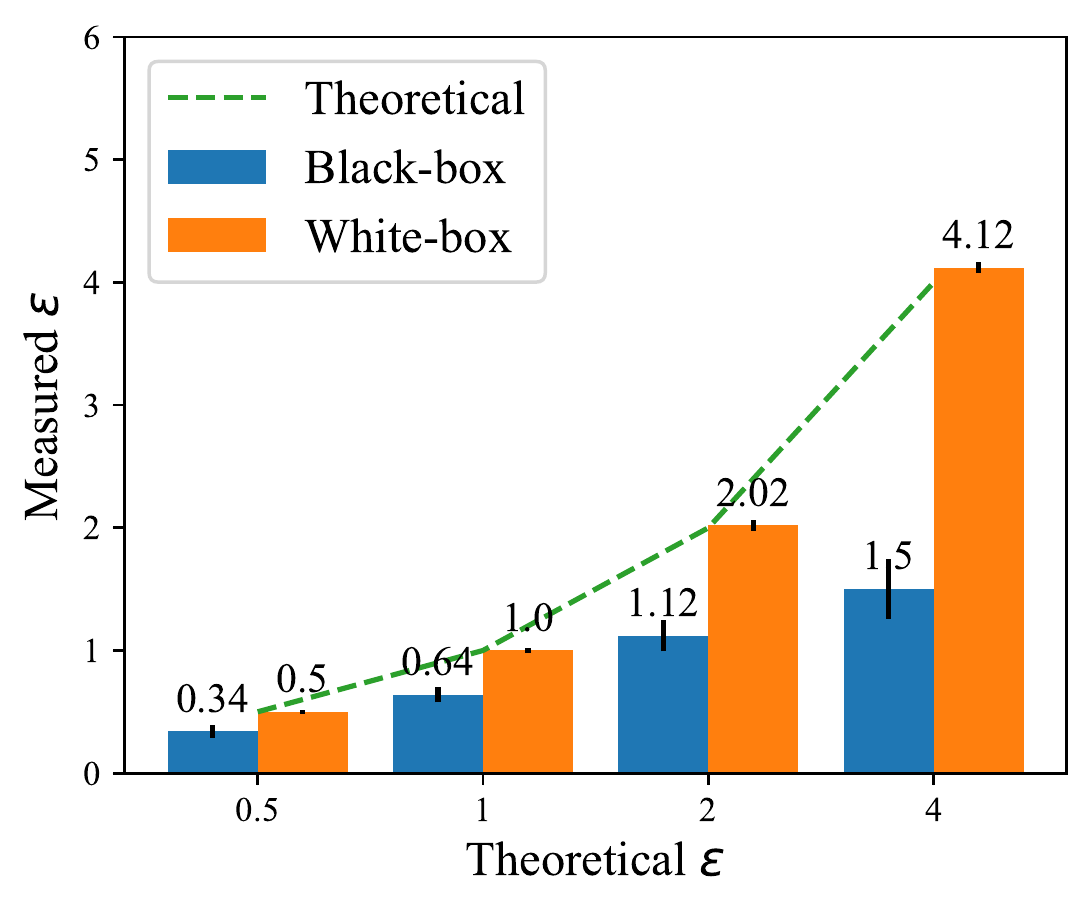}
     \caption{Dummy gradient}
     \label{fig:result_gramax_cifar}
  \end{subfigure}

  \caption{The empirical privacy in federated learning. (CIFAR-10)}
  \label{fig:result-cifar10}
\end{figure}

\begin{figure}[t]
  \centering
  \begin{subfigure}{.29\linewidth}
    \centering
    \includegraphics[width=\linewidth]{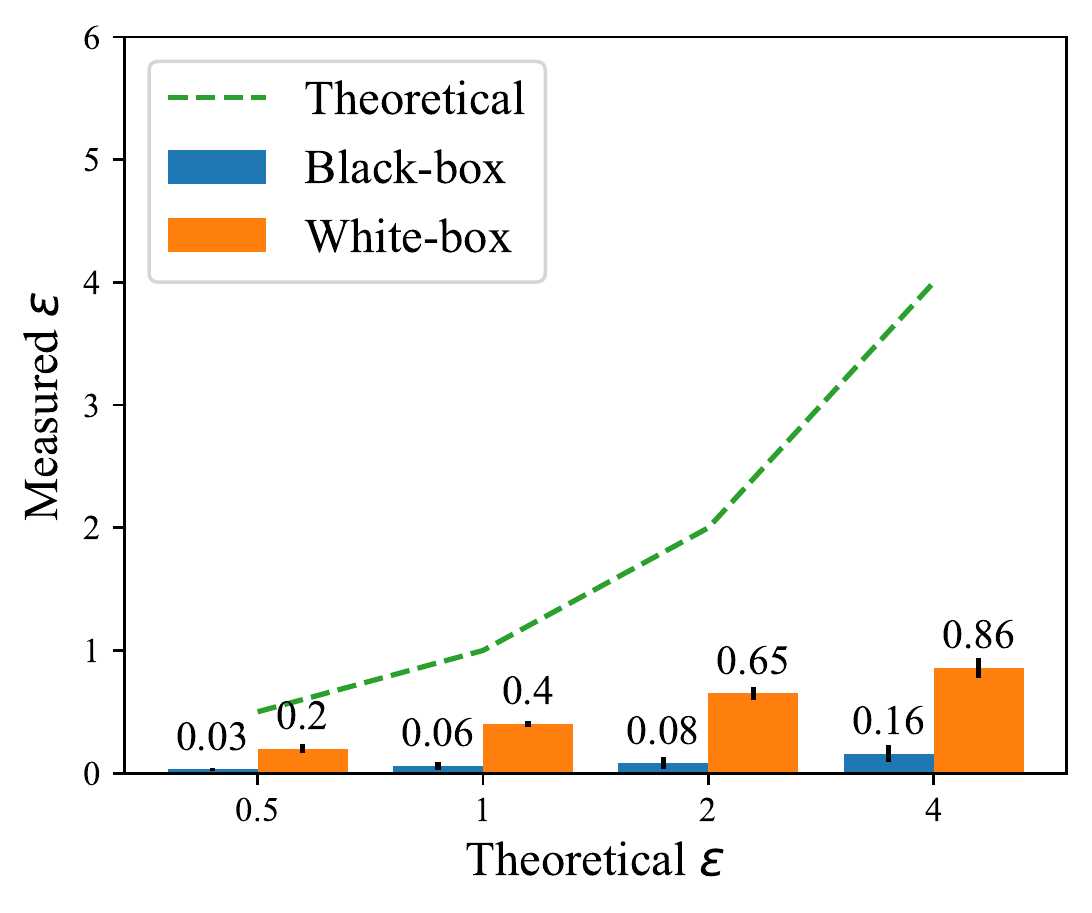}
     \caption{Benign setting}
     \label{fig:benign_setting_fmnist}
  \end{subfigure}
  \begin{subfigure}{.29\linewidth}
    \centering
    \includegraphics[width=\linewidth]{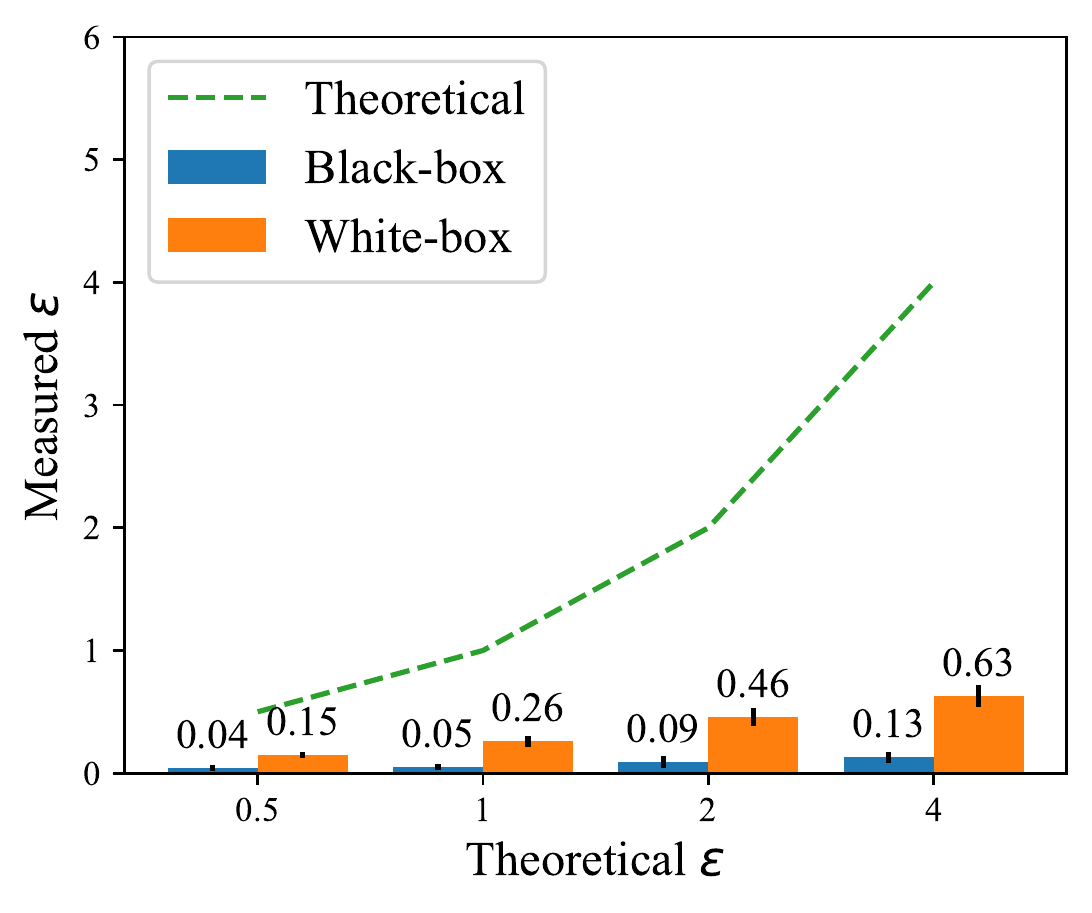}
     \caption{Image perturbation}
     \label{fig:image_perturbation_fmnist}
  \end{subfigure}
  \begin{subfigure}{.29\linewidth}
    \centering
    \includegraphics[width=\linewidth]{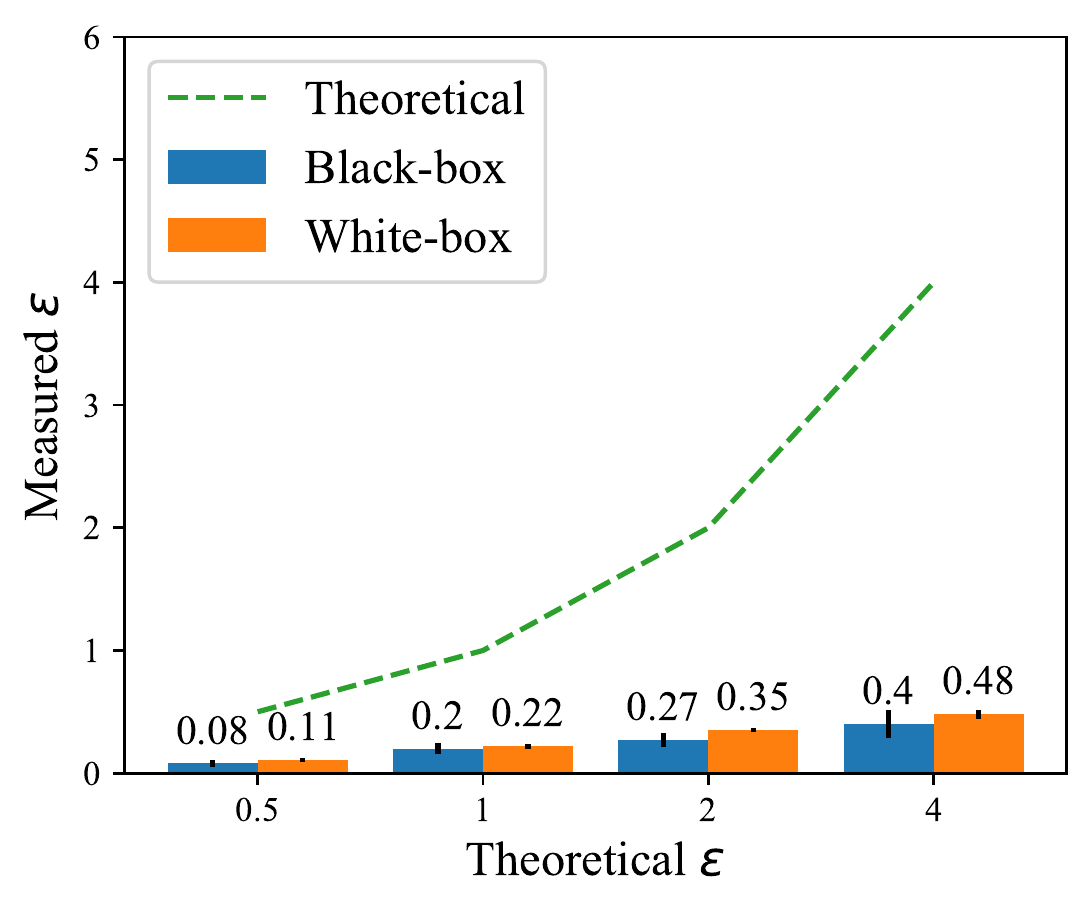}
     \caption{Parameter retrogression}
     \label{fig:param_retrogression_fmnist}
  \end{subfigure}
  
  \begin{subfigure}{.29\linewidth}
    \centering
    \includegraphics[width=\linewidth]{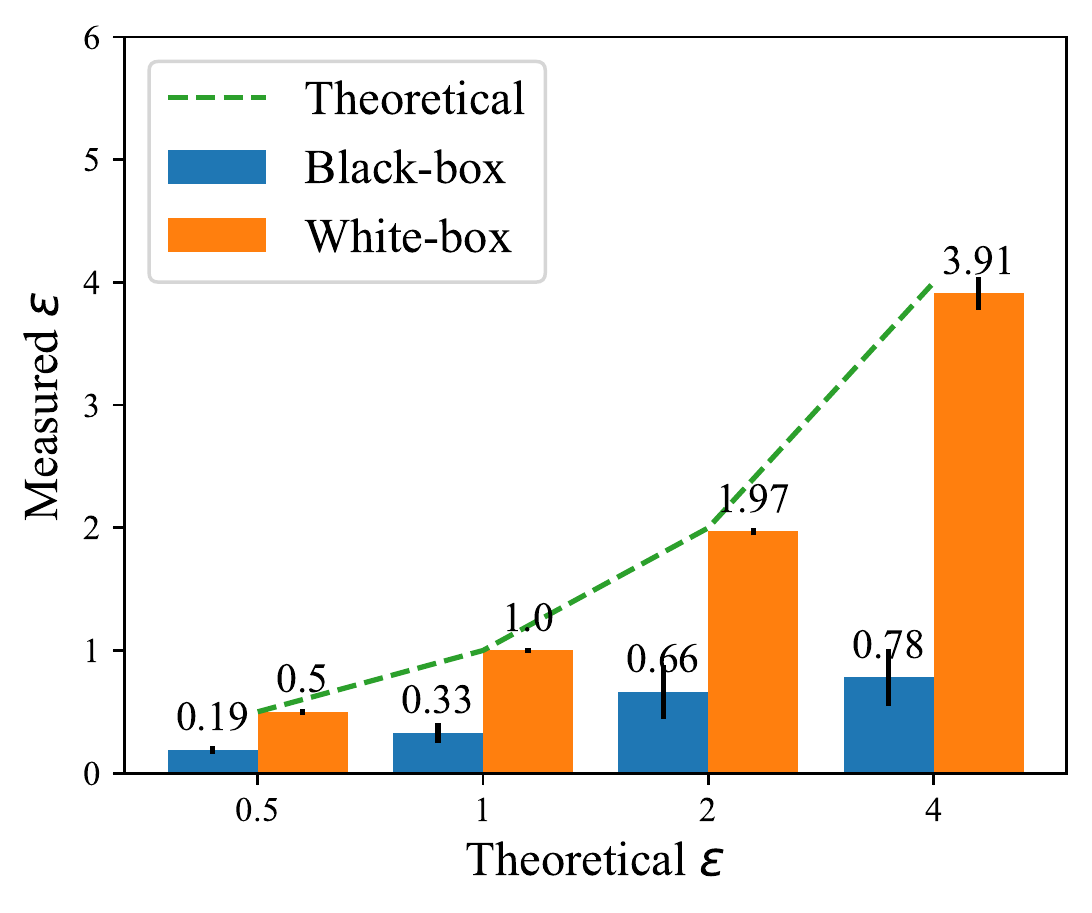}
     \caption{Gradient flip}
     \label{fig:gradflip_fmnist}
  \end{subfigure}
  \begin{subfigure}{.29\linewidth}
    \centering
    \includegraphics[width=\linewidth]{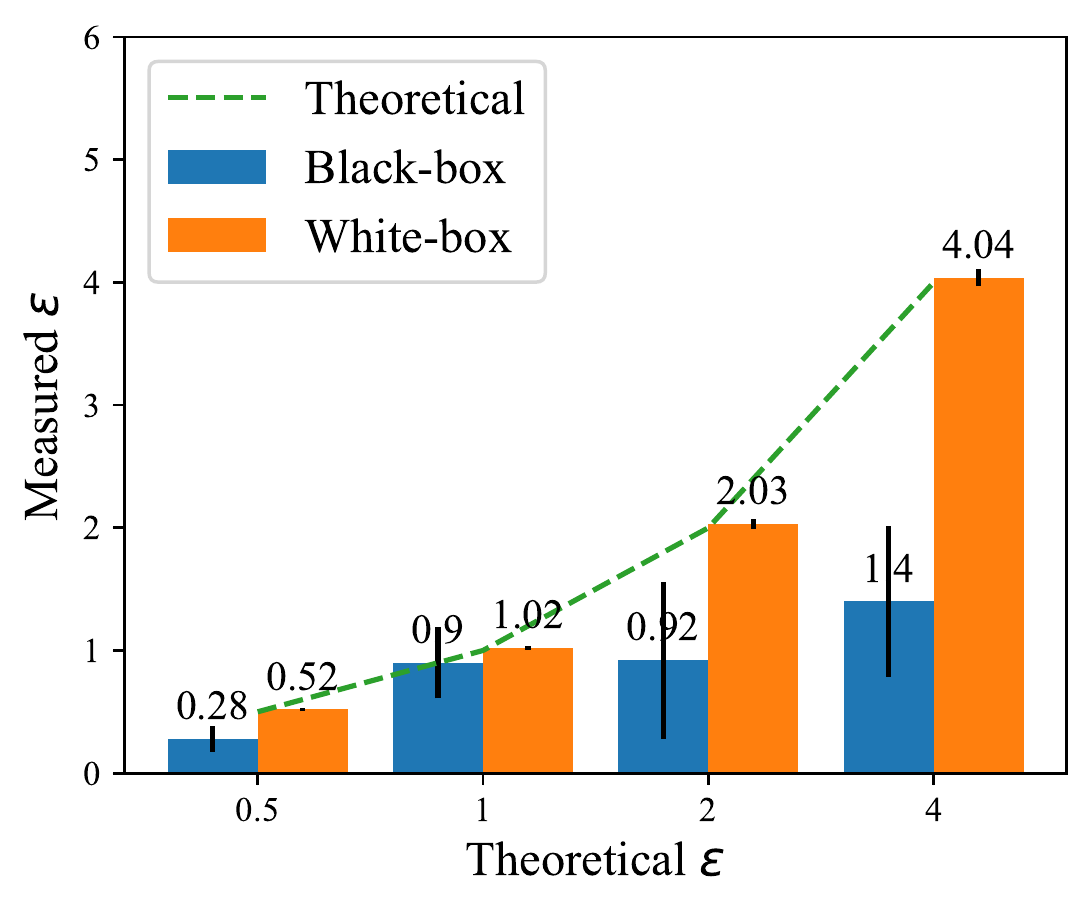}
     \caption{Collusion}
     \label{fig:gradflip_mm_fmnist}
  \end{subfigure}
  \begin{subfigure}{.29\linewidth}
    \centering
    \includegraphics[width=\linewidth]{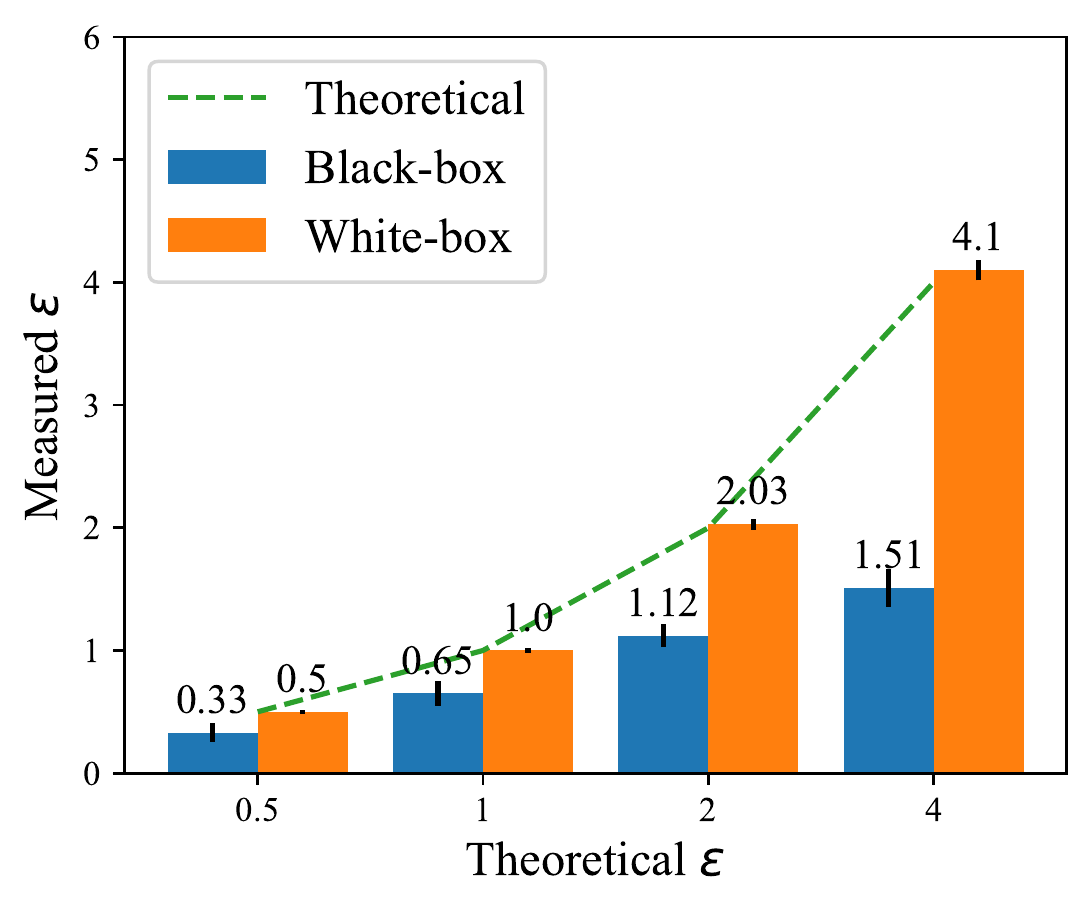}
     \caption{Dummy gradient}
     \label{fig:gradflip_maximize_fmnist}
  \end{subfigure}

  \caption{The empirical privacy in federated learning. (Fashion-MNIST)}
  \label{fig:result-fmnist}
\end{figure}

\begin{figure}[t]
  \centering
  \begin{subfigure}{.29\linewidth}
    \centering
    \includegraphics[width=\linewidth]{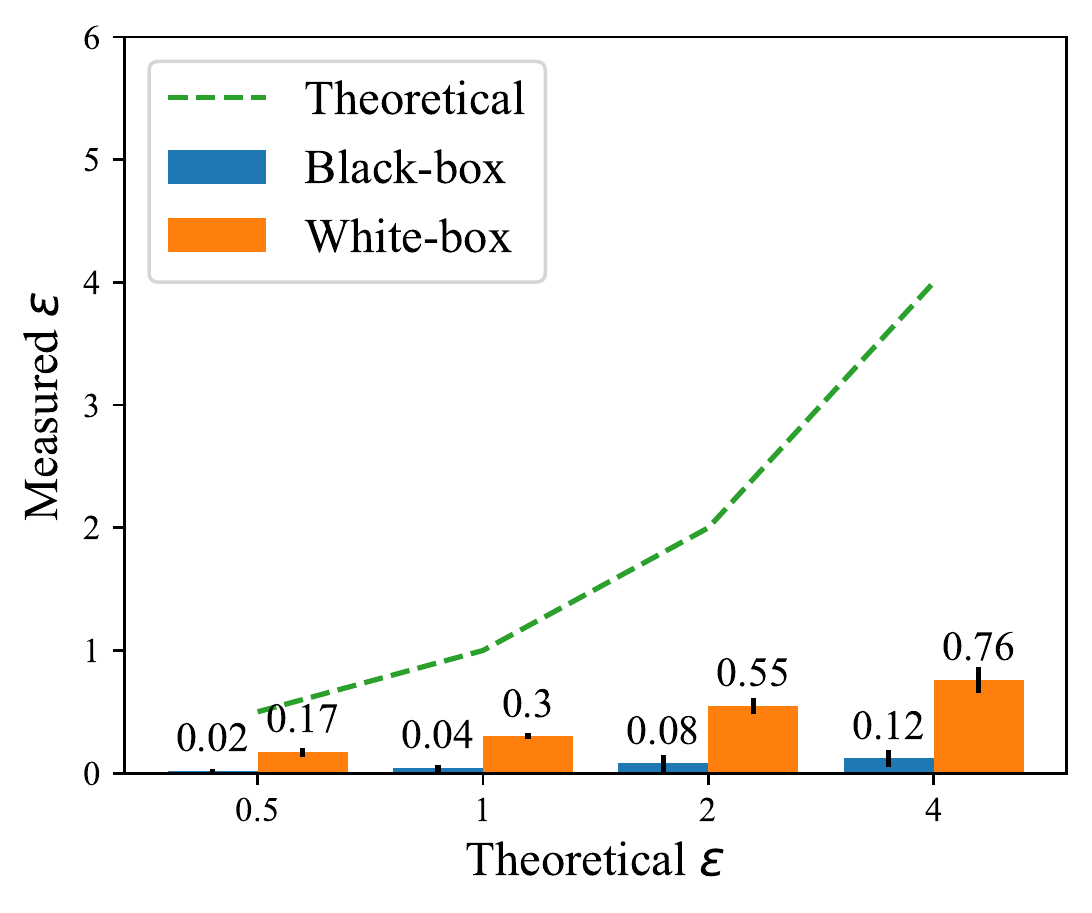}
     \caption{Benign setting}
     \label{fig:benign_setting_svhn}
  \end{subfigure}
  \begin{subfigure}{.29\linewidth}
    \centering
    \includegraphics[width=\linewidth]{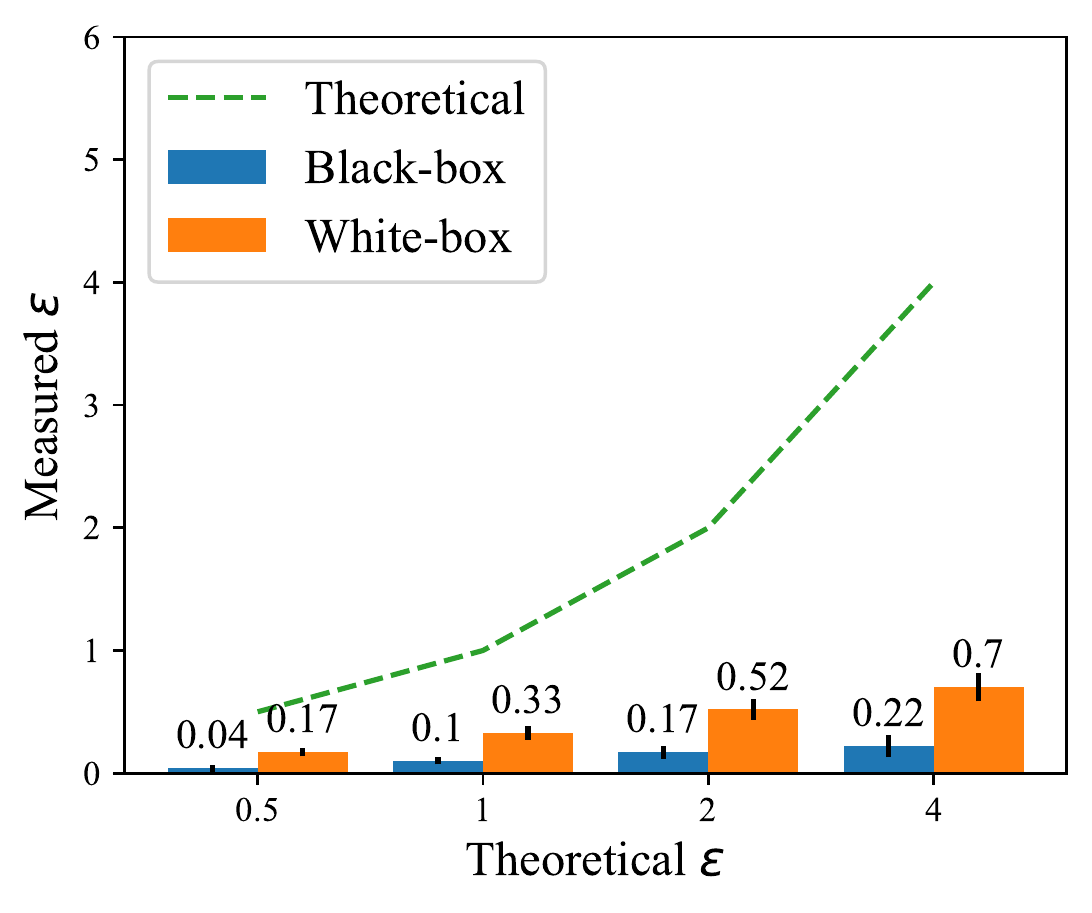}
     \caption{Image perturbation}
     \label{fig:image_perturbation_svhn}
  \end{subfigure}
  \begin{subfigure}{.29\linewidth}
    \centering
    \includegraphics[width=\linewidth]{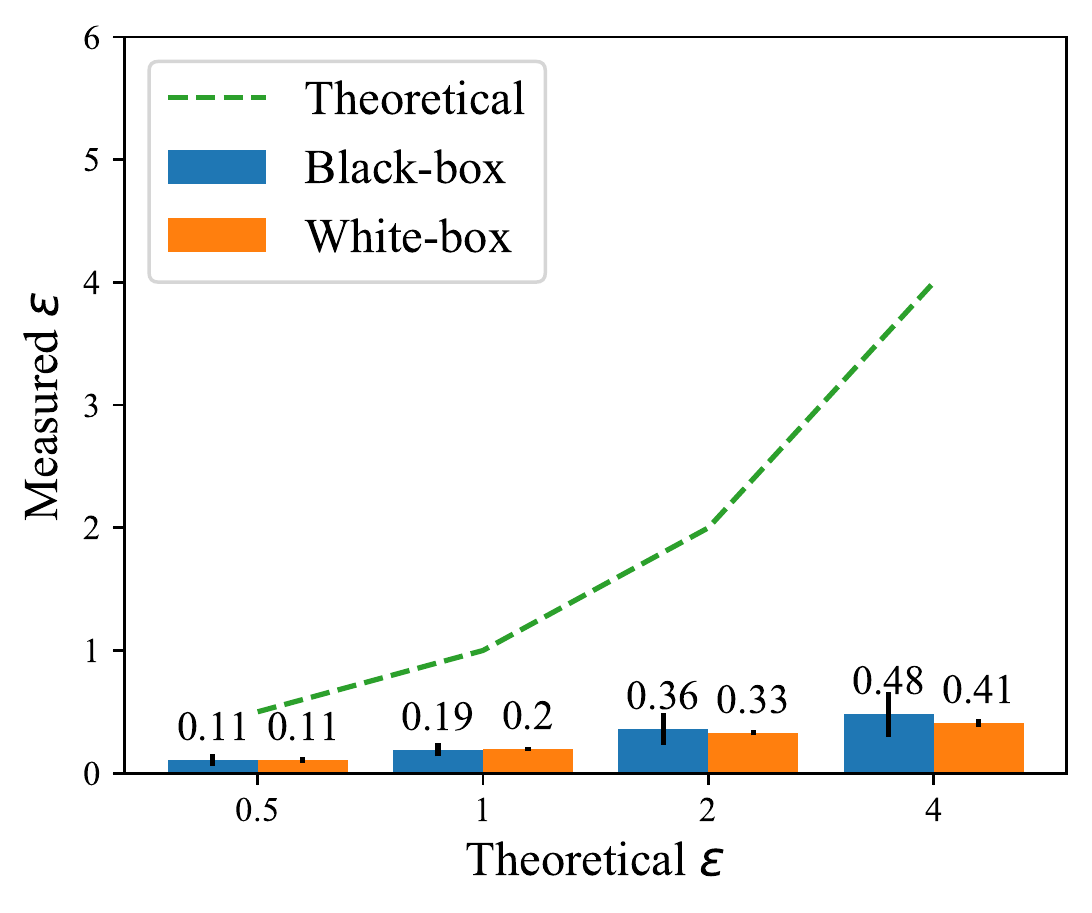}
     \caption{Parameter retrogression}
     \label{fig:param_retrogression_svhn}
  \end{subfigure}
  
  \begin{subfigure}{.29\linewidth}
    \centering
    \includegraphics[width=\linewidth]{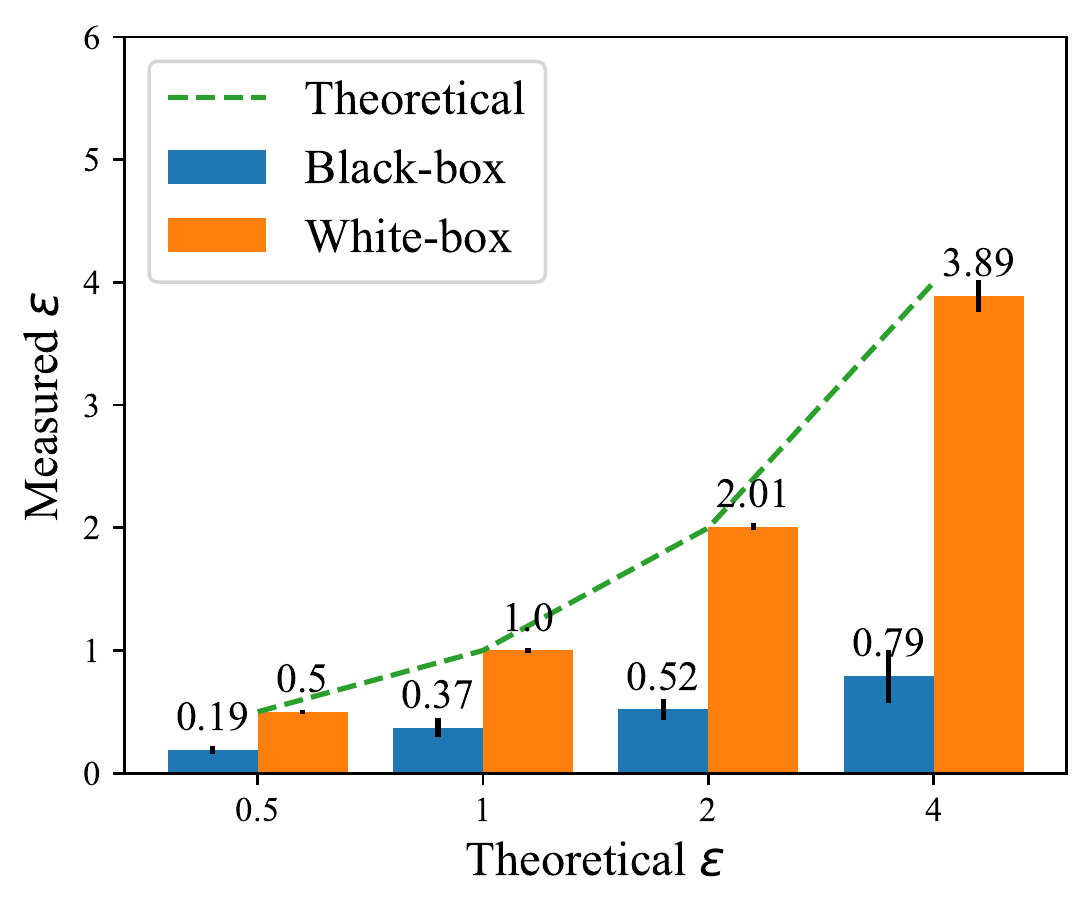}
     \caption{Gradient flip}
     \label{fig:gradflip_svhn}
  \end{subfigure}
  \begin{subfigure}{.29\linewidth}
    \centering
    \includegraphics[width=\linewidth]{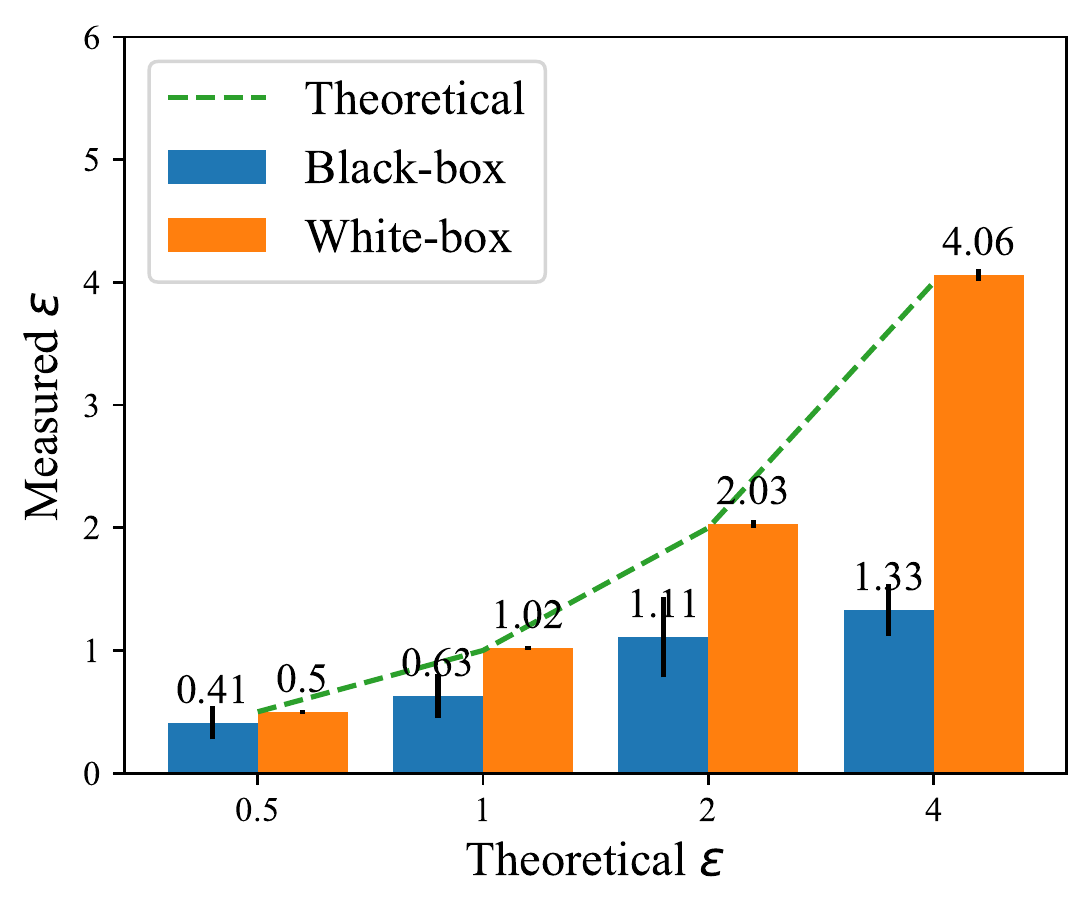}
     \caption{Collusion}
     \label{fig:gradflip_mm_svhn}
  \end{subfigure}
  \begin{subfigure}{.29\linewidth}
    \centering
    \includegraphics[width=\linewidth]{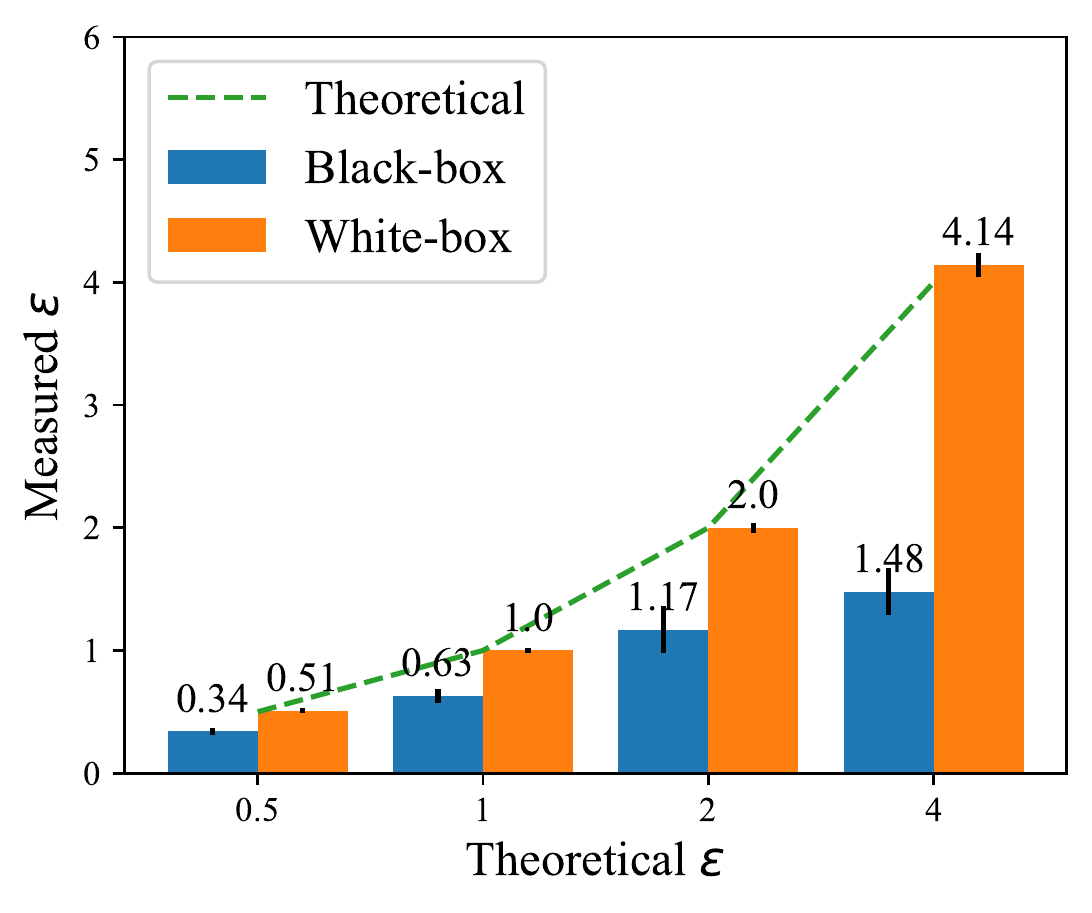}
     \caption{Dummy gradient}
     \label{fig:gradflip_maximize_svhn}
  \end{subfigure}

  \caption{The empirical privacy in federated learning. (SVHN)}
  \label{fig:result-svhn}
\end{figure}


\begin{thebibliography}{43}
\providecommand{\natexlab}[1]{#1}
\providecommand{\url}[1]{\texttt{#1}}
\expandafter\ifx\csname urlstyle\endcsname\relax
  \providecommand{\doi}[1]{doi: #1}\else
  \providecommand{\doi}{doi: \begingroup \urlstyle{rm}\Url}\fi

\bibitem[Abadi et~al.(2016)Abadi, Chu, Goodfellow, McMahan, Mironov, Talwar,
  and Zhang]{dp-sgd-original}
Martin Abadi, Andy Chu, Ian Goodfellow, H~Brendan McMahan, Ilya Mironov, Kunal
  Talwar, and Li~Zhang.
\newblock Deep learning with differential privacy.
\newblock In \emph{Proceedings of the 2016 ACM SIGSAC conference on computer
  and communications security}, pages 308--318, 2016.

\bibitem[Bassily et~al.(2014)Bassily, Smith, and Thakurta]{bassily2014private}
Raef Bassily, Adam Smith, and Abhradeep Thakurta.
\newblock Private empirical risk minimization: Efficient algorithms and tight
  error bounds.
\newblock In \emph{2014 IEEE 55th Annual Symposium on Foundations of Computer
  Science}, pages 464--473. IEEE, 2014.

\bibitem[Bhowmick et~al.(2018)Bhowmick, Duchi, Freudiger, Kapoor, and
  Rogers]{bhowmick2018protection}
Abhishek Bhowmick, John Duchi, Julien Freudiger, Gaurav Kapoor, and Ryan
  Rogers.
\newblock Protection against reconstruction and its applications in private
  federated learning.
\newblock \emph{arXiv preprint arXiv:1812.00984}, 2018.

\bibitem[Bittau et~al.(2017)Bittau, Erlingsson, Maniatis, Mironov, Raghunathan,
  Lie, Rudominer, Kode, Tinnes, and Seefeld]{bittau2017prochlo}
Andrea Bittau, {\'U}lfar Erlingsson, Petros Maniatis, Ilya Mironov, Ananth
  Raghunathan, David Lie, Mitch Rudominer, Ushasree Kode, Julien Tinnes, and
  Bernhard Seefeld.
\newblock Prochlo: Strong privacy for analytics in the crowd.
\newblock In \emph{Proceedings of the 26th Symposium on Operating Systems
  Principles}, pages 441--459, 2017.

\bibitem[Bonawitz et~al.(2017)Bonawitz, Ivanov, Kreuter, Marcedone, McMahan,
  Patel, Ramage, Segal, and Seth]{bonawitz2017practical}
Keith Bonawitz, Vladimir Ivanov, Ben Kreuter, Antonio Marcedone, H~Brendan
  McMahan, Sarvar Patel, Daniel Ramage, Aaron Segal, and Karn Seth.
\newblock Practical secure aggregation for privacy-preserving machine learning.
\newblock In \emph{proceedings of the 2017 ACM SIGSAC Conference on Computer
  and Communications Security}, pages 1175--1191, 2017.

\bibitem[Carlini and Wagner(2017)]{carlini2017adversarial}
Nicholas Carlini and David Wagner.
\newblock Adversarial examples are not easily detected: Bypassing ten detection
  methods.
\newblock In \emph{Proceedings of the 10th ACM workshop on artificial
  intelligence and security}, pages 3--14, 2017.

\bibitem[Costan and Devadas(2016)]{costan2016intel}
Victor Costan and Srinivas Devadas.
\newblock Intel sgx explained.
\newblock \emph{Cryptology ePrint Archive}, 2016.

\bibitem[Duchi et~al.(2018)Duchi, Jordan, and Wainwright]{duchi2018}
John~C Duchi, Michael~I Jordan, and Martin~J Wainwright.
\newblock Minimax optimal procedures for locally private estimation.
\newblock \emph{Journal of the American Statistical Association}, 113\penalty0
  (521):\penalty0 182--201, 2018.

\bibitem[Erlingsson et~al.(2019)Erlingsson, Feldman, Mironov, Raghunathan,
  Talwar, and Thakurta]{erlingsson2019amplification}
{\'U}lfar Erlingsson, Vitaly Feldman, Ilya Mironov, Ananth Raghunathan, Kunal
  Talwar, and Abhradeep Thakurta.
\newblock Amplification by shuffling: From local to central differential
  privacy via anonymity.
\newblock In \emph{Proceedings of the Thirtieth Annual ACM-SIAM Symposium on
  Discrete Algorithms}, pages 2468--2479. SIAM, 2019.

\bibitem[Erlingsson et~al.(2020)Erlingsson, Feldman, Mironov, Raghunathan,
  Song, Talwar, and Thakurta]{ldp-sgd-google}
{\'U}lfar Erlingsson, Vitaly Feldman, Ilya Mironov, Ananth Raghunathan, Shuang
  Song, Kunal Talwar, and Abhradeep Thakurta.
\newblock Encode, shuffle, analyze privacy revisited: Formalizations and
  empirical evaluation.
\newblock \emph{arXiv preprint arXiv:2001.03618}, 2020.

\bibitem[Evfimievski et~al.(2003)Evfimievski, Gehrke, and
  Srikant]{evfimievski2003limiting}
Alexandre Evfimievski, Johannes Gehrke, and Ramakrishnan Srikant.
\newblock Limiting privacy breaches in privacy preserving data mining.
\newblock In \emph{Proceedings of the twenty-second ACM SIGMOD-SIGACT-SIGART
  symposium on Principles of database systems}, pages 211--222, 2003.

\bibitem[Feldman et~al.(2022)Feldman, McMillan, and Talwar]{feldman2022hiding}
Vitaly Feldman, Audra McMillan, and Kunal Talwar.
\newblock Hiding among the clones: A simple and nearly optimal analysis of
  privacy amplification by shuffling.
\newblock In \emph{2021 IEEE 62nd Annual Symposium on Foundations of Computer
  Science (FOCS)}, pages 954--964. IEEE, 2022.

\bibitem[Ford et~al.(2019)Ford, Gilmer, Carlini, and
  Cubuk]{ford2019adversarial}
Nic Ford, Justin Gilmer, Nicolas Carlini, and Dogus Cubuk.
\newblock Adversarial examples are a natural consequence of test error in
  noise.
\newblock \emph{arXiv preprint arXiv:1901.10513}, 2019.

\bibitem[Geiping et~al.(2020)Geiping, Bauermeister, Dr{\"o}ge, and
  Moeller]{paper-inverting-grad}
Jonas Geiping, Hartmut Bauermeister, Hannah Dr{\"o}ge, and Michael Moeller.
\newblock Inverting gradients-how easy is it to break privacy in federated
  learning?
\newblock \emph{Advances in Neural Information Processing Systems},
  33:\penalty0 16937--16947, 2020.

\bibitem[Girgis et~al.(2021)Girgis, Data, Diggavi, Kairouz, and
  Suresh]{girgis2021shuffled}
Antonious Girgis, Deepesh Data, Suhas Diggavi, Peter Kairouz, and
  Ananda~Theertha Suresh.
\newblock Shuffled model of differential privacy in federated learning.
\newblock In \emph{International Conference on Artificial Intelligence and
  Statistics}, pages 2521--2529. PMLR, 2021.

\bibitem[Goodfellow et~al.(2014)Goodfellow, Shlens, and
  Szegedy]{goodfellow2014explaining}
Ian~J Goodfellow, Jonathon Shlens, and Christian Szegedy.
\newblock Explaining and harnessing adversarial examples.
\newblock \emph{arXiv preprint arXiv:1412.6572}, 2014.

\bibitem[Gu et~al.(2017)Gu, Dolan-Gavitt, and Garg]{gu2017badnets}
Tianyu Gu, Brendan Dolan-Gavitt, and Siddharth Garg.
\newblock Badnets: Identifying vulnerabilities in the machine learning model
  supply chain.
\newblock \emph{arXiv preprint arXiv:1708.06733}, 2017.

\bibitem[Hoshino(2020)]{hoshino2020firmfoundation}
Nobuaki Hoshino.
\newblock A firm foundation for statistical disclosure control.
\newblock \emph{Japanese Journal of Statistics and Data Science}, 3, 08 2020.
\newblock \doi{10.1007/s42081-020-00086-9}.

\bibitem[Jagielski et~al.(2020)Jagielski, Ullman, and Oprea]{auditing-DP}
Matthew Jagielski, Jonathan Ullman, and Alina Oprea.
\newblock Auditing differentially private machine learning: How private is
  private sgd?
\newblock \emph{Advances in Neural Information Processing Systems},
  33:\penalty0 22205--22216, 2020.

\bibitem[Jayaraman and Evans(2019)]{236254}
Bargav Jayaraman and David Evans.
\newblock Evaluating differentially private machine learning in practice.
\newblock In \emph{28th USENIX Security Symposium (USENIX Security 19)}, pages
  1895--1912. USENIX Association, 2019.

\bibitem[Kairouz et~al.(2015)Kairouz, Oh, and Viswanath]{paper-kairouz}
Peter Kairouz, Sewoong Oh, and Pramod Viswanath.
\newblock The composition theorem for differential privacy.
\newblock In \emph{International conference on machine learning}, pages
  1376--1385. PMLR, 2015.

\bibitem[Kairouz et~al.(2021)Kairouz, McMahan, Avent, Bellet, Bennis, Bhagoji,
  Bonawitz, Charles, Cormode, Cummings, et~al.]{kairouz2021advances}
Peter Kairouz, H~Brendan McMahan, Brendan Avent, Aur{\'e}lien Bellet, Mehdi
  Bennis, Arjun~Nitin Bhagoji, Kallista Bonawitz, Zachary Charles, Graham
  Cormode, Rachel Cummings, et~al.
\newblock Advances and open problems in federated learning.
\newblock \emph{Foundations and Trends{\textregistered} in Machine Learning},
  14\penalty0 (1--2):\penalty0 1--210, 2021.

\bibitem[Kasiviswanathan et~al.(2011)Kasiviswanathan, Lee, Nissim,
  Raskhodnikova, and Smith]{kasiviswanathan2011can}
Shiva~Prasad Kasiviswanathan, Homin~K Lee, Kobbi Nissim, Sofya Raskhodnikova,
  and Adam Smith.
\newblock What can we learn privately?
\newblock \emph{SIAM Journal on Computing}, 40\penalty0 (3):\penalty0 793--826,
  2011.

\bibitem[Kato et~al.(2022)Kato, Cao, and Yoshikawa]{kato2022olive}
Fumiyuki Kato, Yang Cao, and Masatoshi Yoshikawa.
\newblock Olive: Oblivious and differentially private federated learning on
  trusted execution environment.
\newblock \emph{arXiv preprint arXiv:2202.07165}, 2022.

\bibitem[Krizhevsky et~al.()Krizhevsky, Nair, and Hinton]{alex-cifar-10}
Alex Krizhevsky, Vinod Nair, and Geoffrey Hinton.
\newblock Cifar-10 (canadian institute for advanced research).
\newblock URL \url{http://www.cs.toronto.edu/~kriz/cifar.html}.

\bibitem[LeCun and Cortes(2010)]{lecun-mnisthandwrittendigit-2010}
Yann LeCun and Corinna Cortes.
\newblock {MNIST} handwritten digit database.
\newblock 2010.
\newblock URL \url{http://yann.lecun.com/exdb/mnist/}.

\bibitem[Lee and Clifton(2011)]{lee2011much}
Jaewoo Lee and Chris Clifton.
\newblock How much is enough? choosing $\varepsilon$ for differential privacy.
\newblock In \emph{International Conference on Information Security}, pages
  325--340. Springer, 2011.

\bibitem[Liew et~al.(2022)Liew, Takahashi, Takagi, Kato, Cao, and
  Yoshikawa]{liew2022network}
Seng~Pei Liew, Tsubasa Takahashi, Shun Takagi, Fumiyuki Kato, Yang Cao, and
  Masatoshi Yoshikawa.
\newblock Network shuffling: Privacy amplification via random walks.
\newblock \emph{arXiv preprint arXiv:2204.03919}, 2022.

\bibitem[Liu et~al.(2019)Liu, He, Chanyaswad, Wang, and
  Mittal]{liu2019investigating}
Changchang Liu, Xi~He, Thee Chanyaswad, Shiqiang Wang, and Prateek Mittal.
\newblock Investigating statistical privacy frameworks from the perspective of
  hypothesis testing.
\newblock \emph{Proceedings on Privacy Enhancing Technologies}, 2019:\penalty0
  233--254, 07 2019.
\newblock \doi{10.2478/popets-2019-0045}.

\bibitem[Liu et~al.(2020)Liu, Cao, Yoshikawa, and Chen]{liu2020fedsel}
Ruixuan Liu, Yang Cao, Masatoshi Yoshikawa, and Hong Chen.
\newblock Fedsel: Federated sgd under local differential privacy with top-k
  dimension selection.
\newblock In \emph{International Conference on Database Systems for Advanced
  Applications}, pages 485--501. Springer, 2020.

\bibitem[McKeen et~al.(2013)McKeen, Alexandrovich, Berenzon, Rozas, Shafi,
  Shanbhogue, and Savagaonkar]{mckeen2013innovative}
Frank McKeen, Ilya Alexandrovich, Alex Berenzon, Carlos~V Rozas, Hisham Shafi,
  Vedvyas Shanbhogue, and Uday~R Savagaonkar.
\newblock Innovative instructions and software model for isolated execution.
\newblock \emph{Hasp@ isca}, 10\penalty0 (1), 2013.

\bibitem[McMahan et~al.(2017)McMahan, Moore, Ramage, Hampson, and
  y~Arcas]{mcmahan2017communication}
Brendan McMahan, Eider Moore, Daniel Ramage, Seth Hampson, and Blaise~Aguera
  y~Arcas.
\newblock Communication-efficient learning of deep networks from decentralized
  data.
\newblock In \emph{Artificial intelligence and statistics}, pages 1273--1282.
  PMLR, 2017.

\bibitem[Moosavi-Dezfooli et~al.(2016)Moosavi-Dezfooli, Fawzi, and
  Frossard]{moosavi2016deepfool}
Seyed-Mohsen Moosavi-Dezfooli, Alhussein Fawzi, and Pascal Frossard.
\newblock Deepfool: a simple and accurate method to fool deep neural networks.
\newblock In \emph{Proceedings of the IEEE conference on computer vision and
  pattern recognition}, pages 2574--2582, 2016.

\bibitem[Murakonda and Shokri(2020)]{murakonda2020ml}
Sasi~Kumar Murakonda and Reza Shokri.
\newblock Ml privacy meter: Aiding regulatory compliance by quantifying the
  privacy risks of machine learning.
\newblock \emph{arXiv preprint arXiv:2007.09339}, 2020.

\bibitem[Nasr et~al.(2019)Nasr, Shokri, and Houmansadr]{nasr2019comprehensive}
Milad Nasr, Reza Shokri, and Amir Houmansadr.
\newblock Comprehensive privacy analysis of deep learning: Passive and active
  white-box inference attacks against centralized and federated learning.
\newblock In \emph{2019 IEEE symposium on security and privacy (SP)}, pages
  739--753. IEEE, 2019.

\bibitem[Nasr et~al.(2021)Nasr, Songi, Thakurta, Papemoti, and
  Carlin]{adversary-instantiation}
Milad Nasr, Shuang Songi, Abhradeep Thakurta, Nicolas Papemoti, and Nicholas
  Carlin.
\newblock Adversary instantiation: Lower bounds for differentially private
  machine learning.
\newblock In \emph{2021 IEEE Symposium on Security and Privacy (SP)}, pages
  866--882. IEEE, 2021.

\bibitem[Netzer et~al.(2011)Netzer, Wang, Coates, Bissacco, Wu, and
  Ng]{netzer2011reading}
Yuval Netzer, Tao Wang, Adam Coates, Alessandro Bissacco, Bo~Wu, and Andrew~Y
  Ng.
\newblock Reading digits in natural images with unsupervised feature learning.
\newblock 2011.
\newblock URL \url{http://ufldl.stanford.edu/housenumbers}.

\bibitem[Papernot et~al.(2016)Papernot, McDaniel, Jha, Fredrikson, Celik, and
  Swami]{papernot2016limitations}
Nicolas Papernot, Patrick McDaniel, Somesh Jha, Matt Fredrikson, Z~Berkay
  Celik, and Ananthram Swami.
\newblock The limitations of deep learning in adversarial settings.
\newblock In \emph{2016 IEEE European symposium on security and privacy
  (EuroS\&P)}, pages 372--387. IEEE, 2016.

\bibitem[Papernot et~al.(2017)Papernot, McDaniel, Goodfellow, Jha, Celik, and
  Swami]{papernot2017practical}
Nicolas Papernot, Patrick McDaniel, Ian Goodfellow, Somesh Jha, Z~Berkay Celik,
  and Ananthram Swami.
\newblock Practical black-box attacks against machine learning.
\newblock In \emph{Proceedings of the 2017 ACM on Asia conference on computer
  and communications security}, pages 506--519, 2017.

\bibitem[Song et~al.(2013)Song, Chaudhuri, and Sarwate]{song2013stochastic}
Shuang Song, Kamalika Chaudhuri, and Anand~D Sarwate.
\newblock Stochastic gradient descent with differentially private updates.
\newblock In \emph{2013 IEEE Global Conference on Signal and Information
  Processing}, pages 245--248. IEEE, 2013.

\bibitem[Xiao et~al.(2017)Xiao, Rasul, and Vollgraf]{xiao2017fashion}
Han Xiao, Kashif Rasul, and Roland Vollgraf.
\newblock Fashion-mnist: a novel image dataset for benchmarking machine
  learning algorithms.
\newblock \emph{arXiv preprint arXiv:1708.07747}, 2017.
\newblock URL \url{https://github.com/zalandoresearch/fashion-mnist}.

\bibitem[Yeom et~al.(2018)Yeom, Giacomelli, Fredrikson, and Jha]{paper-MI}
Samuel Yeom, Irene Giacomelli, Matt Fredrikson, and Somesh Jha.
\newblock Privacy risk in machine learning: Analyzing the connection to
  overfitting.
\newblock In \emph{2018 IEEE 31st computer security foundations symposium
  (CSF)}, pages 268--282. IEEE, 2018.

\bibitem[Yin et~al.(2021)Yin, Mallya, Vahdat, Alvarez, Kautz, and
  Molchanov]{paper-cvpr-gradinversion}
Hongxu Yin, Arun Mallya, Arash Vahdat, Jose~M Alvarez, Jan Kautz, and Pavlo
  Molchanov.
\newblock See through gradients: Image batch recovery via gradinversion.
\newblock In \emph{Proceedings of the IEEE/CVF Conference on Computer Vision
  and Pattern Recognition}, pages 16337--16346, 2021.

\end{thebibliography}
\end{document}